\documentclass[11pt]{article}
\usepackage[english]{babel}
\usepackage{graphicx} % Required for inserting images
\usepackage{algorithm}
\usepackage[noend]{algpseudocode}
\usepackage{amsmath}
\usepackage{amsthm}
\usepackage{amssymb}
\usepackage{mathrsfs}
\usepackage{mathalfa}
\usepackage[letterpaper, portrait, margin=1in]{geometry}
\usepackage{marginnote}
\usepackage{thmtools}
\usepackage{thm-restate}
\usepackage{hyperref}
\usepackage{cleveref}
\usepackage{relsize}
\usepackage{bbm}
\usepackage{multirow}
\usepackage{float}
\usepackage{xargs}    
\usepackage[pdftex,dvipsnames]{xcolor}  % Coloured text etc.
\usepackage{relsize}

\makeatletter
\newcommand{\AlignFootnote}[1]{%
	\ifmeasuring@
	\else
	\iffirstchoice@
	\footnote{#1}%
	\fi
	\fi}
\makeatother

\newtheorem{theorem}{Theorem}[section]

\newtheorem{lemma}[theorem]{Lemma}
\newtheorem{remark}[theorem]{Remark}
\newtheorem{claim}[theorem]{Claim}
\newtheorem{proposition}[theorem]{Proposition}
\newtheorem{fact}[theorem]{Fact}
\newtheorem{definition}[theorem]{Definition}

\newcommand{\E}{\mathop{\mathbb{E}}}

\newcommand{\C}{\mathbb{C}}
\newcommand{\GL}{\mathsf{GL}}

\newcommand{\V}[1]{\boldsymbol{#1}}
\newcommand{\trivrep}{\{\mathbf{1}\}}
\newcommand{\gp}{\text{\relsize{-3} $\bullet$}\hspace{2pt}} %groupproduct
\newcommand{\termm}[3]{\mathbf{Term}^{#1}(#2, #3)}

\newcommand{\irr}{\mathsf{Irrep}}
\newcommand{\trace}{\mathsf{tr}}
\newcommand{\comp}{\overline}
\newcommand{\End}[1]{\mathsf{End}\hspace{1pt} #1}
\newcommand{\conv}{*}
\newcommand{\ctranspose}[1]{{#1}^\star}
\newcommand{\hsnorm}[1]{\left\| #1 \right\|_{\mathsf{HS}}}
\newcommand{\opnorm}[1]{\left\| #1 \right\|_{\mathsf{op}}}
\newcommand{\pnorm}[2]{\| #1 \|_{#2}}
\newcommand{\dimgeqi}[2]{\dim_{\geq #2}({#1})}

\newcommand{\commutator}[2]{[#1,#2]}
\newcommand{\N}{\mathbb{N}}

\newcommand{\LC}{{Label Cover}}
\newcommand{\TSAT}{3\text{-}SAT}
\newcommand{\LLC}{{Layered Label Cover}}

\newcommand{\low}{\mathsf{low}}
\newcommand{\high}{\mathsf{high}}

\renewcommand{\geq}{\geqslant}
\renewcommand{\leq}{\leqslant}

\newcommand{\eps}{\varepsilon}

\newcommand{\np}{\mathbf{NP}}
\newcommand{\p}{\mathbf{P}}

\renewcommand{\inf}{\mathrm{Inf}}

\title{Optimal Inapproximability of Generalized Linear Equations over a Finite Group}

\newcommand{\anonymoysflag}{0} % change to 0 to switch behavior

\ifnum\anonymoysflag=1
    \author{Anonymous}
\else
    \author{Amey Bhangale\thanks{Department of Computer Science and Engineering, University of California, Riverside. Supported by the Hellman Fellowship award and NSF CAREER award 2440882.}
    \and
    Yezhou Zhang\thanks{Department of Computer Science and Engineering, University of California, Riverside. Supported by the Hellman Fellowship award and NSF CAREER award 2440882.}
    }
\fi

\date{}

\begin{document}
	\allowdisplaybreaks
	\maketitle

\begin{abstract}
    Constraint satisfaction problems (CSPs) consist of a set of variables taking values from some finite domain and a set of local constraints on these variables. The objective is to find an assignment to the variables that maximizes the fraction of satisfied constraints.
    
    In this work, we study the CSP where the constraints are {\em generalized linear equations} over a finite group $G$. More specifically, for a given $S\subseteq G$, the constraints in this CSP are of the form addition of the values to the variables (similarly, product for non-abelian groups) belongs to the set $S$. We give an approximation algorithm for this problem on satisfiable instances and show that it is optimal for certain $S$ assuming $\mathbf{P}\neq \mathbf{NP}$. %To show this optimal hardness result, we design a novel decoding procedure in the soundness analysis of the reduction, which may be of independent interest towards understanding the optimal approximability of satisfiable CSPs.
    
    This natural predicate is one of the very few known predicates that are approximation resistant on almost satisfiable instances, assuming $\mathbf{P}\neq \mathbf{NP}$, but admits a non-trivial approximation algorithm on satisfiable instances. 
\end{abstract}	

\section{Introduction}

Constraint satisfaction problems  (CSPs) are one of the most fundamental problems in theoretical computer science. A Max-$P$-CSP instance $\varphi$ for a given predicate $P: \Sigma^k \rightarrow \{0,1\}$ consists of a set of variables $x_1, x_2, \ldots, x_n$ taking values from the domain $\Sigma$ and a collection of constraints $C_1, C_2, \ldots, C_m$ where each $C_i$ consists of a constraint of the form $P(x_{i_1}, x_{i_2}, \ldots, x_{i_k})$. The constraints might involve {\em literals} instead of just the variables. The objective is to assign values to the variables that maximize the fraction of satisfied constraints from $\varphi$. An $\alpha$-approximation for Max-$P$-CSP is an algorithm that, given an instance $\varphi$ of Max-$P$-CSP, outputs an assignment that satisfies at least $\alpha\cdot \mathrm{OPT}$, where $\mathrm{OPT}$ is the optimal value of the instance $\varphi$.

A systematic study of the complexity of  CSPs was started by Schaefer in 1978~\cite{Schaefer78}, who showed that for every predicate $P$ over a $2$-element set, the problem $P$-CSP is either solvable in polynomial time or is $\np$-complete.  A famous dichotomy conjecture of Feder and Vardi~\cite{FederV1,FederV2}, which was resolved recently in huge breakthroughs by Bulatov~\cite{Bulatov17} and Zhuk~\cite{Zhuk20} independently, states that for every predicate $P$, checking the satisfiability of a $P$-CSP is either in $\p$ or is $\np$-complete.  However, when it comes to approximation algorithms for Max-$P$-CSP, the question is wide open.

The optimal approximability results for Max-$P$-CSP for various predicates $P$ are known starting with the seminal work of~\cite{Hastad2001}. H{\aa}stad showed that Max-3SAT, where the predicate is over a Boolean alphabet and is an OR of three literals, is $\np$-hard to approximate within a factor of $\frac{7}{8}+\eps$ for every $\eps>0$. The above hardness result also holds even if the given instance is guaranteed to have an assignment with value $1$. It is also easy to get a $\frac{7}{8}$-approximation algorithm for Max-3SAT -- for each variable, pick a random value from $\{0,1\}$. This algorithm satisfies $\frac{7}{8}$-fraction of the constraints in expectation. It is also easy to derandomize this algorithm by the method of conditional expectation. Thus, the hardness result is optimal. %Such optimal hardness results are known for various predicates, including Max Set Splitting~\cite{Hastad2001}, Max-NTW~\cite{Hastad14}, a predicate that contains a subgroup that is balanced pairwise independent~\cite{Chan16}, etc.

In the same work on the optimal hardness of Max-3SAT, H{\aa}stad~\cite{Hastad2001} also studied the inapproximability of Max-3LIN over an abelian group $G$. In Max-3LIN over a group $G$, the variables take values from the group $(G, +)$, and constraints are of the type $x_{i_1} + x_{i_2} + x_{i_3} = c$ for some $c\in G$. If the given instance of Max-3LIN is fully satisfiable, then it is known~\cite{RG99} that the satisfying assignment can be found in polynomial time using techniques similar to Gaussian elimination. Furthermore, a random assignment to the variables satisfies $\frac{1}{|G|}$ fraction of the constraints in expectation. H{\aa}stad~\cite{Hastad2001} showed that this algorithm is optimal for general instances. In other words, he showed that for every $\eps>0$, if an instance of Max-3LIN is given with value at least $1-\eps$, then it is $\np$-hard to find an assignment that satisfies at least $\frac{1}{|G|}+\eps$ fraction of the constraints. 

The Max-3LIN over a non-abelian group $(G, \gp)$ is another interesting problem. Unlike the abelian group case, here, Goldmann and Russell showed~\cite{RG99} that it is $\np$-complete to check the satisfiability of a given instance for every non-abelian group $G$. Engebretsen, Holmerin, and Russell gave~\cite{EngebretsenHR04} similar inapproximability results as in the case of abelian groups. More specifically, they showed that for every $\eps>0$, if an instance of Max-3LIN over a non-abelian group $(G, \gp)$ is given with value at least $1-\eps$, then it is $\np$-hard to find an assignment that satisfies at least $\frac{1}{|G|}+\eps$ fraction of the constraints. There is, however, a better than $\frac{1}{|G|}$ approximation algorithm for certain groups on satisfiable instances. It is folklore to get a $\frac{1}{|\commutator{G}{G}|}$-approximation algorithm for Max-3LIN over a non-abelian group $(G, \gp)$ on satisfiable instances, where $\commutator{G}{G}$ is the commutator subgroup of $G$. A commutator of two group elements $g$ and $h$ is a group element $g^{-1}\gp h^{-1}\gp g\gp h$, and a commutator subgroup $\commutator{G}{G}$ is the subgroup generated by all the commutators of the group $G$. Recently, Bhangale and Khot~\cite{BhangaleK21} showed that this algorithm is optimal on satisfiable instances assuming $\p \neq \np$.\\

Given these approximation algorithms and inapproximability results concerning Max-3LIN, it is natural to ask what role the abelian nature of a group $G$ plays in such results. In this work, we study a generalized version of the Max-3LIN problem, denoted by $\text{Max-E$k$-LIN}_{S}(G)$, as defined below.\\

	Fix a group $(G, \gp)$ and a subset $S\subseteq G$. In the instance of $\text{Max-E$k$-LIN}_{S}(G)$, the variables take values from $G$. A {\em literal} of a variable $x$ is given by $(g\gp x)$ for some $g\in G$.
	\begin{definition}[$\text{Max-E$k$-LIN}_{S}(G)$]
		Consider a multiset of linear equations $(C_1,C_2,\dots, C_m)$ over variables $\{x_1,x_2,\dots,x_n|x_i\in G\}$ for a group $G$. Each constraint consists of a tuple of exactly $k$ literals and will be considered satisfied when their sum/product over $G$ is in $S\subseteq G$. The $\text{Max-E$k$-LIN}_{S}(G)$ problem is to find an assignment to $\{x_1,x_2,\dots,x_n\}$ that maximizes the number of satisfied constraints.
	\end{definition}

Observe that the problem Max-3LIN over a group $G$ that was discussed earlier is  $\text{Max-E$k$-LIN}_{\{1_G\}}(G)$ where $1_{G}$ is the identity element of the group $G$. In $\text{Max-E$k$-LIN}_{S}(G)$, we allow more satisfying assignments in the predicate, i.e, ignoring the literals for simplicity, the condition that a tuple $(a, b, c)$ satisfies a given constraint depends on the value of $a\gp b\gp c$, and hence we call the predicate a {\em generalized linearity predicate} over $G$.\footnote{This nomenclature was also used by Chattopadhyay and Wigderson~\cite{ChattopadhyayW09} to describe similar predicates.}

\paragraph{Approximation Algorithm for $\text{Max-E$3$-LIN}_{S}(G)$.} To keep things simple, we restrict to the setting when $G$ is an abelian group for this discussion. We start with a simple approximation algorithm for $\text{Max-E$3$-LIN}_{S}(G)$. Let $\Phi$ be an instance of $\text{Max-E$3$-LIN}_{S}(G)$ with constraints set $(C_1, C_2, \ldots, C_m)$ over the variables $X = \{ x_1, x_2, \ldots, x_n\}$. Here, the constraint $C_i$ is of the form
$$(a_{i_1}+x_{i_1}) + (a_{i_2}+ x_{i_2}) + (a_{i_3}+x_{i_3}) \in S.$$
One can try to replace checking this condition by a linear equation over a certain abelian group, hoping to solve this system using Gaussian elimination, and from this solution get a non-trivial solution to the original instance. One natural way to do this is to consider the group $(Q,+):= G/\mathrm{cl}(S)$ where $\mathrm{cl}(S)$ is the subgroup generated by $S$. The above condition implies that

\begin{equation}
    \label{eq:abelian_eq_rhs_1}
([a_{i_1}]_Q+ [x_{i_1}]_Q) + ([a_{i_2}]_Q+ [x_{i_2}]_Q) + ([a_{i_3}]_Q+ [x_{i_3}]_Q) = 0_Q,
\end{equation}%
where $[z]_Q$ for $z\in G$ denotes the element of $Q$, which are the cosets of $\mathrm{cl}(S)$, where the respective coset contains $z$, and $0_Q$ is the identity element of $Q$. Thus, we convert the set of constraints in $\Phi$ to a system of equations, denoted by $\tilde{\Phi}$, over the group $(Q,+):= G/\mathrm{cl}(S)$ with variables $Y = \{ y_1, y_2, \ldots, y_n\}$. More specifically, for a constraint $C_i$ which is of the form  $(a_{i_1}+x_{i_1}) + (a_{i_2}+ x_{i_2}) + (a_{i_3}+x_{i_3}) \in S$, we add the following equation over $Q$: $[a_{i_1}]_Q + y_{i_1} + [a_{i_2}]_Q + y_{i_2} + [a_{i_3}]_Q + y_{i_3} = 0_Q$.

Now, as $\Phi$ is satisfiable, consider the satisfying assignment $\V \alpha: X\rightarrow G$ to $\Phi$. Consider the assignment $\tilde{\V \alpha} : Y\rightarrow Q$ given by the natural map $\tilde{\V \alpha}(y_i) = [\V \alpha(x_i)]_Q$ . It is easy to see that $\tilde{\V \alpha}$ satisfies all the equations from the instance $\tilde{\Phi}$, and hence, $\tilde{\Phi}$ is satisfiable. Since $\tilde{\Phi}$ is a system of equations over an abelian group $(Q,+)$, we can find a satisfying assignment to $\tilde{\Phi}$ in polynomial time using Gaussian elimination~\cite{RG99}. Let $\tilde{\V \beta}$ be the assignment returned by this procedure. To construct the final assignment to the $X$ variables, we simply set $x_i$ to be a random element from the coset $\tilde{\V \beta}(y_i)$. Let $\V \beta: X\rightarrow G$ be the random assignment given by the above procedure. It can be easily observed that $\V \beta$ satisfies a given constraint $C_i$ in $\Phi$ with probability $\frac{|S|}{|\mathrm{cl}(S)|}$ and hence $\V \beta$ satisfies $\frac{|S|}{|\mathrm{cl}(S)|}$ fraction of the constraints in expectation. The randomized algorithm can be easily derandomized using the method of conditional expectations.

The above algorithm is indeed optimal for certain pairs $(S, G)$. However, consider a group $G = \mathbb{Z}_4 \times \mathbb{Z}_4$ and $S = \{(0,1), (1,0)\}$. In this case, $\mathrm{cl}(S) = G$ and hence the above algorithm gives a $|S|/|\mathrm{cl}(S)| = 1/8$ approximation algorithm. However, there is a $1/2$-approximation algorithm for this CSP as described next. 

One simple observation is that in Equation (\ref{eq:abelian_eq_rhs_1}), the RHS being the identity element of the group was not important for the efficient solvability of the system of equations over $Q$. Thus, one can work with a potentially smaller subgroup $H_S$ such that $S\subseteq g+H_S$ for some $g\in G$. In this case, letting $Q = G/H_S$, the equations in $\tilde{\Phi}$ are of the form $[a_{i_1}]_Q + y_{i_1} + [a_{i_2}]_Q + y_{i_2} + [a_{i_3}]_Q + y_{i_3} = g$, and the same rounding procedure as above gives a $|S|/|H_S|$-approximation guarantee.

For the earlier example of  $G = \mathbb{Z}_4 \times \mathbb{Z}_4$ and $S = \{(0,1), (1,0)\}$, one can take the subgroup $H_S = \{ (0,0), (1, 3), (3, 1), (2,2) \}$ where $S\subseteq (0,1)+H_S$, which gives a $\frac{1}{2}$-approximation.\\

Can this approximation guarantee be further improved? Our main result shows that the above approximation algorithm (and a generalization of this to non-abelian groups) is indeed optimal for certain $S$.
% \abnote{our proof needs $\langle S^{-1}S\rangle$ generating $H_S$ for now}
	
	\begin{theorem}
    \label{thm:main}
		Fix any finite group $G$ and $S\subseteq G$. Let $H_S$ is the smallest normal subgroup such that
		\begin{enumerate}
			\item $[G,G]\subseteq H_S$
			\item $S\subseteq gH_S$ for some $g\in G$, i.e., $S$ is a subset of some coset of $H_S$.
		\end{enumerate}
         then the problem $\text{Max-E$k$-LIN}_{S}(G)$ is approximable within a factor of $\frac{|S|}{|H_S|}$ on satisfiable instances. Furthermore, if $S^{-1}S$ generates $H_S$, assuming $\p\neq \np$, for every $\eps>0$, it is $\np$-hard to approximate $\text{Max-E$k$-LIN}_{S}(G)$ within a factor of $\frac{|S|}{|H_S|}+\eps$.
	\end{theorem}

Note that for any abelian group $G$, the subgroup $[G,G]$ consists only of the identity element of $G$. For non-abelian groups, the condition $[G,G]\subseteq H_S$ guarantees that the quotient group $Q:=G/H_S$ is abelian, which was crucial for the above-discussed approximation algorithm. See Theorem~\ref{thm:approx_alg} for a straightforward generalization of the above approximation algorithm for non-abelian groups. 

The above theorem shows that even for abelian groups $G$, the problem $\text{Max-E$k$-LIN}_{S}(G)$ is $\np$-hard on satisfiable instances, in general. 

\begin{remark}
    We remark that on instances that are almost satisfiable, picking a random assignment, that gives $\frac{|S|}{|G|}$-approximation for $\text{Max-E$k$-LIN}_{S}(G)$, is optimal assuming $\p\neq \np$. This follows easily by modifying the hardness reductions from~\cite{Hastad2001} for abelian groups and from~\cite{EngebretsenHR04} for non-abelian groups.
\end{remark}

A predicate $P$ is called approximation resistant if it is $\np$-hard to do better than the random assignment algorithm. We are aware of two results, one by H{\aa}stad~\cite{Hastad13} on satisfying degree-$d$ equations over GF$[2]^n$, and another by Bhangale-Khot~\cite{BhangaleK21} on Max-3LIN over non-abelian groups $G$ where $G$ is not a {\em simple} group. Both these predicates are approximation resistant in general (i.e., on almost satisfiable instances), but have non-trivial approximation algorithms on satisfiable instances. Thus, our result adds the generalized linear equation predicate to the small class of predicates that are known to be approximation resistant in general but have non-trivial approximation algorithms on satisfiable instances. We hope that our result gives another piece of information to help understand the approximability of Max-CSPs on satisfiable instances.

In Section~\ref{sec:approx_alg}, we give a simple $\frac{|S|}{|H_S|}$-approximation algorithm for $\text{Max-E$k$-LIN}_{S}(G)$ for every finite group $G$. The algorithm implicitly uses the abelian embedding of the generalized linearity predicate as defined in the series of work~\cite{BKMI, BKMII, BKMIV, BKMV} towards understanding the approximability of satisfiable CSPs. The main technical part of our work is to show that the above approximation algorithm is optimal for the certain subset $S$ assuming $\p\neq \np$. To show this optimal hardness result, we design a novel decoding procedure in the soundness analysis of the reduction, which may be of independent interest towards understanding the approximability of satisfiable CSPs. More specifically, the soundness analysis uses the fact that doing the Gaussian elimination over a certain abelian group is useful towards getting a better approximation algorithm for this problem on satisfiable instances. We elaborate more on this in the techniques section below.

\subsection{Related Work}
In this section, we go over relevant work on the inapproximability of Max-CSPs. The PCP Theorem~\cite{AroraLMSS1998, AroraS1998, FGLSS96} shows that Max-$P$-CSPs are $\np$-hard to approximate within a factor of $(1-\delta)$ for some constant $\delta>0$ if checking satisfiability of $P$-CSP is $\np$-complete. H{\aa}stad in his seminal work~\cite{Hastad2001} greatly improved the hardness of approximation results for a few CSPs. Notable examples of the CSPs from his work include Max-$3$SAT and Max-$3$LIN. For Max-CUT, which is a $2$-ary CSP, H{\aa}stad showed that it is $\np$-hard to approximate within a factor of $\frac{16}{17}$. Khot~\cite{Khot02} formulated the {\em Unique Games Conjecture} (UGC), which is a conjecture on the hardness of the Label Cover instances (see Definition~\ref{def:label-cover}) restricted to the constraints being $1$-to-$1$. Samorodnitsky and Trevisan~\cite{SamorodnitskyT00} showed that Boolean Max-$k$CSP (with arbitrary $k$-ary Boolean predicates) is $\np$-hard to approximate beyond $O(\frac{k}{2^k})$ assuming the UGC, matching the best algorithm up to a constant factor. Khot, Kindler, Mossel, and O’Donnell~\cite{KKMO07} showed that Max-CUT is $\np$-hard to approximate within a factor of $0.878$ assuming the UGC, which matches the approximation guarantee of the Goemans-Williamson~\cite {GW95} algorithm for Max-CUT.

Raghavendra~\cite{Rag09} presented an elegant result that generalizes the above Max-CUT reduction and establishes, for any Max-$P$-CSP instance, the $(c, s)$-integrality gap of the basic Semidefinite programming relaxation implies finding an $s+\eps$ satisfying assignment on $c-\eps$ satisfiable instances is $\np$-hard assuming the UGC. This result fully characterizes the approximability of Max-$P$-CSPs assuming the UGC on almost satisfiable instances. Furthermore, given a fixed predicate $P$, Raghavendra's result does not explicitly give the optimal hardness factor for Max-$P$-CSP. Austrin and Mossel~\cite{AustrinM09} gave the right threshold for predicates that support a uniform and pairwise independent distribution. A distribution $\mu$ on $P^{-1}(1)\subseteq \Sigma^k$ is said to be pairwise independent if for every distinct pair $i, j\in [k]$, the marginal of $\mu$ restricted to the coordinates $(i,j)$ is uniform over $\Sigma^2$. Austrin and Mossel showed that such predicates are approximation resistant on almost satisfiable instances, assuming the UGC.

Chan~\cite{Chan16} established a general criterion for approximation resistance, resolving the NP-hardness of Max-$k$-CSP up to a constant factor and assuming $\p\neq\np$. Specifically, he proved the hardness for Max-CSPs where the domain is an abelian group $G$ and the predicate $P^{-1}(1)\subseteq G^k$ is a subgroup that satisfies a condition analogous to that identified by Austrin and Mossel.

The question of finding the optimal approximation algorithm (even assuming certain conjectures, like $d$-to-$1$ conjecture~\cite{Khot02} or Rich $2$-to-$1$ conjecture~\cite{BravermanKM21}) for satisfiable instances of Max-$P$-CSPs is wide open. In a recent series of work, Bhangale, Khot, and Minzer~\cite{BKMI, BKMII} defined a property of abelian embeddability of the predicate towards understanding the approximability of satisfiable CSPs. A predicate $P: \Sigma^k \rightarrow \{0,1\}$ is said to have an abelian embedding in an abelian group $G$, if there are maps $\alpha_i: \Sigma\rightarrow G$, not all constant, such that $\sum_i \alpha_i(a_i) = 0_G$ for every $(a_1, a_2, \ldots, a_k) \in P^{-1}(1)$. They gave an optimal dictatorship test for $3$-ary predicates that have no abelian embedding. Very recently, for certain $3$-ary predicates that have an abelian embedding, they gave an approximation algorithm~\cite {BKMIV, BKMV} for satisfiable instances that uses a combination of Gaussian elimination as well as the SDP rounding algorithm. They also showed the `optimality' of this algorithm by giving a dictatorship test with matching parameters.

\subsection{Techniques}
\label{sec:techniques}
In this section, we give an overview of the techniques in the inapproximability results of our main theorem. 

\subsubsection{Abelian Groups}
We begin with the case of abelian groups to highlight one of the main differences in the analysis of the reduction compared to the seminal work of H{\aa}stad on Max-3LIN over an abelian group. We assume some familiarity with the Fourier analysis of functions over abelian groups (for instance, Chapter 8 of Ryan O'Donnell's book~\cite{O14}). Throughout the section, $\eps >0$ is an arbitrarily small constant.

\newcommand{\pr}{P}

Starting with a work of H{\aa}stad~\cite{Hastad2001}, a typical way to prove the hardness of approximation is to start with an $\np$-hard problem called the {\em \LC}~\cite{GHS02, Khot02a, Khot02b, DinurG13, KhotS13, Hastad14, GuruswamiHMRC11}. A gadget is built on top of the \LC instance to create an instance of a given CSP. For simplicity of the presentation, we focus here is on the most important and technical component of the reduction. This component is the construction of {\em dictatorship tests} and analyzing the tests.

A function $f:[q]^n \rightarrow [q]$ is called a dictator function if $f(x_1, x_2, \ldots, x_n) = x_j$ for some $j\in [n]$. A dictatorship test for a predicate $\pr:[q]^k \rightarrow \{0,1\}$ consists of a distribution $\mu$ that is (almost) supported on the set of satisfying assignments of $\pr$. The test samples $k$ inputs $z_1, z_2, \ldots, z_k$ as follows: For each coordinate $i\in [n]$, the tuple $((z_1)_i, (z_2)_i, \ldots, (z_k)_i)$  is sampled independently from $\mu$. The test accepts $f$ if $\pr(f(z_1), f(z_2), \ldots, f(z_k))$ evaluates to $1$, i.e., $f(z_1), f(z_2), \ldots, f(z_k)$ forms a satisfying assignment for $\pr$. It is clear that if $f$ is a dictator function, then the test passes with probability (almost) $1$. This is because, in this case, we are checking if the $i^{th}$ coordinate of the inputs is from $\pr^{-1}(1)$, which is always true by construction. 

Once the distribution is fixed, the next step is to analyze the soundness of the dictatorship test, i.e., the probability with which the test passes if $f$ is `far from the dictator functions'. The notion of far from dictator functions changes based on the hardness reduction. A typical notion that is used is that the function has all the variables with degree $d = O(1)$ influences low. The influence of the $i^{th}$ coordinate on the function is the probability that, on a random input, changing the $i^{th}$ coordinate changes the values of the function. In terms of the Fourier coefficients of $f$, this is equal to the following quantity:
	$$\inf_i(f):= \sum_{\alpha: \alpha_i\neq 0} |\hat{f}(\alpha)|^2.$$
Thus, the $i^{th}$ dictator function has $\inf_i(f) = 1$. A degree-$d$ influence of the $i^{th}$ variable is given by the following expression,
	$$\inf_i^{\leq d}(f):= \sum_{\alpha: \alpha_i\neq 0 \wedge |\alpha|\leq d} |\hat{f}(\alpha)|^2,$$
where $|\alpha|$ is the number of non-zero coordinates of $\alpha$. With this notion of far from dictator functions in mind, the dictatorship test used in the hardness reduction~\cite{Hastad2001} of Max-3LIN works as follows.

\begin{enumerate}
	\item Select $\V x, \V y \sim G^n$ uniformly at random.
	\item Set $\V z = \V x+\V y$.
    \item For each $i\in [n]$, resample $(x_i, y_i, z_i)$ from $G^3$ uniformly at random, with probability $\eps$. 
	\item Check if $f(\V x) + f(\V y) = f(\V z) $.
\end{enumerate}
It is clear that any dictator function passes the above test with probability at least $1-\eps$. To analyze the soundness of the test, we can express the test passing probability as follows:
		\begin{align*}
			\Pr[\mbox{Test passes}]  = \frac{1}{|G|}\E_{ (\V x, \V y, \V z)} \left[ \sum_{\chi_\rho\in \hat{G}} \chi_{\rho}(f(\V x)+ f(\V y)+f(\V z)) \right],
		\end{align*}
        where the summation is over all the characters of the group $G$. The term with $\chi_\rho$ being the trivial character gives $\frac{1}{|G|}$. For the remaining terms with $\chi_\rho$ being a non-trivial character, we are left with analyzing the following expectation
        $$\E_{ (\V x, \V y, \V z)} \left[  \chi_{\rho}(f(\V x)+ f(\V y)+f(\V z)) \right] = \E_{ (\V x, \V y, \V z)} \left[\chi_{\rho}(f(\V x))\cdot \chi_\rho(f(\V y))\cdot \chi_\rho(f(\V z)) \right].$$
        By letting $F(\V w) = \chi_\rho(f(\V w))$ and expanding the function $F$ with the Fourier basis over $G^n$, and doing some simplifications, we get that the above expectation is upper bounded as follows. 
        \begin{align*}
        \E_{ (\V x, \V y, \V z)} \left[\chi_{\rho}(f(\V x))\cdot \chi_\rho(f(\V y))\cdot \chi_\rho(f(\V z)) \right] &\leq \sum_{\alpha} |\hat{F}(\alpha)|^3 \cdot (1-\eps)^{|\alpha|}\\
        & =  \sum_{\alpha \wedge |\alpha|\leq d} |\hat{F}(\alpha)|^3 \cdot (1-\eps)^{|\alpha|} + \sum_{\alpha \wedge |\alpha|> d} |\hat{F}(\alpha)|^3 \cdot (1-\eps)^{|\alpha|}.
        \end{align*}
       \noindent{\bf Decoding Strategy.} Now, because of the `{\em noise}' (i.e., step 3 in the test), the second term can be shown to be negligible for some large $d$ as a function of $\eps$, by a simple application of Cauchy-Schwarz inequality and using Parseval's identity. As for the first term, if it is non-negligible, then the structure of the function can be used in the hardness reduction starting from the \LC instance.
       
       More precisely, the starting point of the reduction is a \LC instance (see Definition~\ref{def:label-cover}). It is $\np$-hard to distinguish between the \LC instances with value $1$ from instances with value at most $\delta$ for small $\delta>0$. In the reduction, each vertex $v$ of the \LC instance is replaced with a cloud of vertices $C[v]$ where $|C[v]|$ is $G^M$ where $M$ is the label size of the vertex. These constitute the variables/literals of the reduced instance of Max-3LIN. The distribution on the constraints is specified by the above dictatorship test. If the value of the \LC instance is $1$, then the value of the reduced Max-3LIN instance is at least $1-\eps$. Now, similar to the above analysis, if we fix an assignment $f$ to the Max-3LIN instance with value $\frac{1}{|G|} + \eps$ for some $\eps>0$, then the assignment restricted to most of $C[v]$, call it  $f_v$, have high degree-$d$ influential variables. From this, one can come up with a labeling to the \LC instance with value greater than $\delta = \delta(\eps)$, thereby showing the soundness of the reduction. The strategy simply picks a random non-trivial $\rho$, a Fourier coefficient $\widehat{\chi_\rho(f_v)}(\alpha)$ of $f_v$ with probability $|\widehat{\chi_\rho(f_v)}(\alpha)|^2$ and assign a label $i$ such that $\alpha_i$ is not a trivial character.\\

        We now give the natural extension of the above dictatorship test to the generalized linear equation predicate that we consider in this paper. Towards this, fix an abelian group $(G, +)$ and a subset $S\subseteq G$. Consider the following dictatorship test.

\begin{enumerate}
	\item Select $\V x, \V y \sim G^n$ uniformly at random.
    \item Select $\V w\sim S^n$ uniformly at random.
    \item Set $\V z = - \V x-\V y + \V w$.
    \item Check if $f(\V x) + f(\V y)  + f(\V z)\in S $.
\end{enumerate}

It is clear that every dictator function passes the test with probability $1$. Similar to the above analysis, the soundness of the test can be expressed as follows:
		\begin{align*}
			\Pr[\mbox{Test passes}]  = \frac{1}{|G|}\E_{ (\V x, \V y, \V z)} \left[ \sum_{s\in S} \sum_{\chi_\rho\in \hat{G}} \chi_{\rho}(f(\V x)+ f(\V y)+f(\V z) - s) \right],
		\end{align*}
        Once again, the terms that correspond to the trivial character give $\frac{|S|}{|G|}$ (which corresponds to the approximation ratio of the algorithm that picks a random assignment). Note that our goal is to show that the test passes with probability almost $\frac{|S|}{|H_S|} + \eps$ in the soundness towards proving Theorem~\ref{thm:main}. Once again, the term that corresponds to a given $s\in S$ and a character $\chi_\rho$ gives,
          $$\E_{ (\V x, \V y, \V z)} \left[  \chi_{\rho}(f(\V x)+ f(\V y)+f(\V z) -s) \right] = \E_{ (\V x, \V y, \V z)} \left[\chi_{\rho}(f(\V x))\cdot \chi_\rho(f(\V y))\cdot \chi_\rho(f(\V z)) \cdot\chi_\rho(s)\right].$$
Ignoring the constant shift $\chi_\rho(s)$, we are again left with analyzing the expectation 
$$\E_{ (\V x, \V y, \V z)} \left[\chi_{\rho}(f(\V x))\cdot \chi_\rho(f(\V y))\cdot \chi_\rho(f(\V z))\right].$$
Recall that the subgroup $H_S$ is the smallest subgroup such that $S\subseteq g+H_S$ for some $g\in G$. Now, unlike the Max-3LIN, we cannot expect this expectation to be small when the functions $F$ have negligible degree-$d$ influences for some $d= O(1)$. To see this, consider a character $\chi_\rho$ that is constant on the subgroup $H_S$. If we let $f(\V x) = \sum_i x_i$ where the sum is the group operation, then the derived function $F$ has all degree-$d$ influences $0$. However, the expectation becomes $\E[\chi_\rho(\sum_i w_i)]$, since $S\subseteq g+H_S$ for some $g$, we get that the $\sum_i w_i \in  ng H_S$, i.e, it always belongs to a fixed coset of $H_S$. Since the character $\chi_\rho$ is constant on the subgroup $H_S$ (and hence constant on every coset of $H_S$), we get that the expectation is $1$ in absolute value. The number of such characters (including the trivial character) for which we cannot bound the expectation is precisely $|G|/|H_S|$. This gives the right factor $\frac{|S|}{|H_S|}$ in the test passing probability that we need for our Theorem~\ref{thm:main}.

Now, consider the character $\chi_\rho$, which is not constant on $H_S$. Similar to H{\aa}stad's analysis, we can upper bound the corresponding expectation as follows
      \begin{align*}
        \E_{ (\V x, \V y, \V z)} \left[\chi_{\rho}(f(\V x))\cdot \chi_\rho(f(\V y))\cdot \chi_\rho(f(\V z)) \right] &\leq \sum_{\alpha} |\hat{F}(\alpha)|^3 \cdot (1-\eta)^{|\alpha|_{S}},
        \end{align*}
where $\eta>0$ is a non-zero constant that only depends on $|G|$ and $|\alpha|_{S}$ is the number of coordinates of $\alpha$ where the character $\alpha_i$ is {\em not constant} on the subgroup $H_S$. In order to work with this expression, we modify the notion of far from dictator functions that will be useful for our hardness reduction as follows.

\paragraph{\bf Modified low-degree influences.} A modified degree-$d$ influence of the $i^{th}$ variable is expressed as the following expression,
	$$\inf_i^{\leq d}(f):= \sum_{\alpha: \alpha_i\neq 0 \wedge |\alpha|_{S}\leq d} |\hat{f}(\alpha)|^2,$$
where $|\alpha|_{S}$ is the number of coordinates of $\alpha$ where the character $\alpha_i$ is not constant on the subgroup $H_S$. With this change, we split the summation as follows:
   \begin{align*}
        \sum_{\alpha} |\hat{F}(\alpha)|^3 \cdot (1-\eta)^{|\alpha|_{S}}
        & =  \sum_{\alpha \wedge |\alpha|_{S}\leq d} |\hat{F}(\alpha)|^3 \cdot (1-\eta)^{|\alpha|_{S}} + \sum_{\alpha \wedge |\alpha|_{S}> d} |\hat{F}(\alpha)|^3 \cdot (1-\eta)^{|\alpha|_{S}}.
    \end{align*}
Again, the second term is negligible for some large $d=O_{|G|}(1)$. \\

\noindent{\bf Modified Decoding Strategy.} In our analysis of the reduction, we need to show that the terms that are similar to the first term are not negligible, then there is a decoding strategy (similar to the one described in the reduction to Max-3LIN above) in the hardness reduction starting from the \LC instance. The following modified strategy works. The strategy picks a random non-trivial $\rho$ that is not constant on $H_S$, a Fourier coefficient $\widehat{\chi_\rho(f_v)}(\alpha)$ of $f_v$ with probability $|\widehat{\chi_\rho(f_v)}(\alpha)|^2$ and assign a label $i$ such that the character corresponding to $\alpha_i$ is not constant on $H_S$.\\

\subsubsection{Non-abelian Groups}
The above dictatorship test can be easily modified for the non-abelian case in a natural way
% (see Section~\ref{sec:nonabelian_red} for more details). 
\begin{enumerate}
    \item Select $\V x,\V y\sim G^n$ uniformly at random.
    \item Select $\V s\in S^n $ uniformly at random.
    \item For each $i\in [n]$, set $z_i = y_i^{-1}\gp x_i^{-1}\gp s_i$.
    \item Check if $f(\V x)\gp f(\V y)\gp f(\V z)\in S$.
\end{enumerate}
The completeness case is trivial. In the soundness case, we again express the test passing probability as 
	\begin{align*}
			\Pr[\mbox{Test passes}]   = \frac{1}{|G|}\E_{(\V x, \V y, \V z)} \left[ \sum_{s\in S} \left[ \sum_{\rho\in \irr(G)}\dim(\rho)\cdot \chi_{\rho}(f(\V x)\gp f(\V y)\gp  f(\V z)\gp s^{-1}) \right] \right],
		\end{align*}
where the inner sum is over all irreducible representations of $(G, \gp)$ and $\dim(\rho)$ is the dimension of the representation $\rho$. Similarly to the abelian case, for certain representations of dimension $1$ (that are constant on the subgroup $H_S$ defined in Theorem~\ref{thm:main}), we bound the expectation by $1$ in absolute value. This gives the factor $\frac{|S|}{|H_S|}$. The analysis for the remaining dimension $1$ representations, Lemma \ref{lemma:soundness_dimone}, is analogous to the one described in the abelian case but focuses on the natural structure of non-abelian groups and generalized linear equation predicates. For some technical reasons, for non-abelian groups, we could analyze this under the assumption that $S^{-1} S$ generates the subgroup $H_S$.

Regarding representations with dimension $\geq 2$, we treat the analysis of Bhangale and Khot~\cite{BhangaleK21} as a black box (Lemma~\ref{lemma:highterm_bound}) to conclude that the associated expectations are small unless they yield a decoding strategy for the \LC\ instance. For some technical reasons, the analysis of this part of the reduction requires the use of \LLC\ instead of the bipartite \LC. Therefore, in our main reduction, we also use a \LLC\ instance as a starting point, but we only use the layered version in the proof of Lemma~\ref{lemma:highterm_bound} to adopt results from~\cite{BhangaleK21new}. For the primary part of this paper, we only need a bipartite \LC\ instance.

\begin{remark}
    The conference version of the paper~\cite{BhangaleK21} had an error, which the authors fixed with the use of \LLC\ as a starting point.\footnote{personal communication~\cite{BhangaleK21new}} This fix is reflected in the Lemma~\ref{lemma:bk_new_high}~\cite[Claim 4.5]{BhangaleK21new} that we use as black-box.
\end{remark}

\subsection{Organization}
 We begin Section~\ref{sec:prelims} by defining the \LC\ instance and the hardness of approximation of \LC, which is the starting point of our reduction. In Section~\ref{sec:reptheory}, we go over the basics of Fourier analysis over general finite groups. We formally give the approximation algorithm described in the introduction in Section~\ref{sec:approx_alg}. Finally, in Section~\ref{sec:nonabelian_red} we give our hardness reduction for $\text{Max-E$3$-LIN}_{S}(G)$ and analyze the reduction. The hardness for $\text{Max-E$k$-LIN}_{S}(G)$ for $k\geq4$ follows easily from a similar reduction, thereby proving the main Theorem~\ref{thm:main}.

\section{Preliminaries}
        \label{sec:prelims}

	\subsection{Label Cover}

	We start by defining the \LC\ and \LLC\ problem, which we use as a starting point for our reduction.
	\begin{definition}[\LC]
		
		\label{def:label-cover} An instance $\Psi=(U,V,E,[L],
		[R],\{\pi_e\}_{e\in E})$ of the {\LC} constraint satisfaction
		problem consists of a bi-regular bipartite graph $(U,V,E)$,
		   alphabets $[L]$ and $[R]$ and a surjective projection map $\pi_e : [L] \rightarrow
		[R]$ for every edge $e\in E$. 
		Given a labeling $\ell : U \rightarrow [L], \ell:V \rightarrow
		[R]$, an edge $e = (u,v)$ is said to be satisfied by $\ell$ if
		$\pi_e(\ell(u)) = \ell(v)$. 
		
		$\Psi$ is said to be \emph{satisfiable} if there exists a 
		labeling that satisfies all the
		edges. $\Psi$ is said to be \emph{at most $\delta$-satisfiable} if every
		labeling satisfies at most a $\delta$ fraction of the
		edges. 
	\end{definition}
	
	The hardness of \LC\ stated below follows from
	the PCP Theorem \cite{AroraS1998,AroraLMSS1998, FGLSS96} and Raz's Parallel
	Repetition Theorem \cite{Raz1998}. %The additional structural property on the hard instances (item 2 below) is proved by H\aa stad \cite[Lemma 6.9]{Hastad2001}.
	\begin{theorem}[Hardness of \LC]
		
		\label{thm:LChard} For every $r \in \N$, there is a deterministic $n^{O(r)}$-time reduction from a \TSAT\ instance
		of size $n$ to an instance $\Psi=(U,V,E,[L],[R], \{\pi_e\}_{e\in E})$ of \LC\ with the following properties: $|U|,|V| \leq n^{O(r)}$; $L,R \leq 2^{O(r)}$; $\Psi$ is bi-regular with degrees bounded by $2^{O(r)}$.
			
			\begin{itemize}
				
				\item YES Case : If the \TSAT\ instance is
				satisfiable, then $\Psi$ is satisfiable.
				
				\item NO Case : If the \TSAT\ instance is
				unsatisfiable, then $\Psi$ is
				at most $2^{-r}$-satisfiable.
			\end{itemize}
	\end{theorem}
For our hardness result, we need the following variant of the \LC~problem.

	\begin{definition}[\LLC]
	    An $T$-\LLC\ instance, given by $\mathcal{H}=(\mathcal{V}=\{V_1,\dots,V_T\}, \{\Pi_{i,j}\}_{1\leq i < j\leq T}\}, \{[R_i]_{i\in [T]}\}$ consist of $T$ sets of vertices $\mathcal{V}=\{V_1,\dots,V_T\}$. The label set of vertices in layer $i$ is denoted by $[R_i]$. Every pair of layers $1\le i<j\le l$ has a set of constraints $\Pi_{i,j}$ between the vertices in $V_{i}$ and $V_{j}$. The constraint between $v\in V_{i}$ and $u\in V_{j}$ (if it exists in $\Pi_{i,j}$) is denoted by $\pi_{vu}$. Moreover, every constraint between a pair of vertices is a projection constraint: for every assignment $k\in[R_{i}]$ to $v\in V_{i}$, there is a unique assignment to $u\in V_{j}$ that satisfies the constraint $\pi_{vu}$. 
	\end{definition}
    \begin{theorem}[\cite{DinurGKR03}, Hardness of \LLC]
    \label{thm:layered_lc_hardness}
        For any constant parameters $T\ge2$, $r\in\mathbb{Z}$, the following problem is NP-hard. Given an $T$-\LLC\ instance $\mathcal{H}=(\mathcal{V}=\{V_{1},...,V_{T}\},\{\Pi_{i,j}\}_{1\le i<j\le T},\{[R_{t}]\}_{t\in[T]})$ where all variable ranges $R_{t}$ are of size $2^{O(Tr)}$, distinguish between the following two cases:

    \begin{itemize}
        \item \textbf{Completeness.} There is an assignment satisfying all the constraints of the Label Cover instance. In this case, we say that $\mathcal{H}$ is fully satisfiable.
        \item \textbf{Soundness.} For every $1\le i<j\le T$, no assignment satisfies more than a $2^{-r}$ fraction of the set of constraints $\Pi_{i,j}$ between layers $i$ and $j$. In this case, we say that $\mathcal{H}$ is at most $2^{-r}$-satisfiable.
    \end{itemize}
    \end{theorem}

	\subsection{Fourier analysis}

	In this section, we give a brief overview of the representation theory of non-abelian groups and the Fourier analysis of non-abelian groups. For a more comprehensive understanding, we refer the reader to the book by Terras~\cite{Terras99}. We state many propositions in the following subsection, and the proofs of these propositions can be found in the same book \cite{Terras99}. This section is identical to the section from~\cite{BhangaleK21}.
	
	\subsubsection{Representation Theory}
	\label{sec:reptheory}
	In this paper, we only consider non-abelian groups which are {\em finite}. Let $G = (G,\gp)$ be a finite non-abelian group. The identity element of a group is denoted by $1_G$.
	
	\begin{definition}
		A representation $(V, \rho)$ of $G$ is a vector space $V$ together with a group homomorphism $\rho: G \rightarrow \GL(V)$ from $G$ to the group $\GL(V)$ of invertible $\C$-linear transformations from $V$ to $V$. The dimension of the vector space $V$ is denoted by $\dim(\rho)$.
	\end{definition} 
	
	For convenience, we just use the letter $\rho$ to denote a representation of $G$ and use $\rho_V$ to denote the underlying vector space. We view a representation $\rho(\cdot)$ as its corresponding matrix of the linear transformation. Thus $\rho(\cdot)_{ij}$ is used to denote the $(i,j)^{th}$ entry of that matrix. We always work with representations that are unitary. There is one representation that is obvious -- just map everything to $1\in \C$.  This representation is called the {\em trivial representation}, which has dimension $1$. We will denote the trivial representation by $\trivrep$. 
	
	\begin{definition}
		Let $\rho$ and $\tau$ be representations of $G$. An isomorphism from $\rho_V$ to
		$\tau_V$ is an invertible linear transformation $\phi : \rho_V \rightarrow \tau_V$ such that 
		$$ \phi \circ \rho(g) = \tau(g)\circ \phi,$$
		for all $g\in G$. We say that $\rho_V$ and $\tau_V$ are isomorphic and write $\rho_V\cong \tau_V$ if there exists an isomorphism from $\rho_V$ to $\tau_V$.
	\end{definition}

	\begin{definition}
		Let $\rho$ be a representation of $G$. A vector subspace $W\subset \rho_V$ is $G$-invariant if $\rho(g)w \in W$ for all $g\in G$ and $w\in W$.
	\end{definition}
	
	If a representation $(V, \rho)$ has a $G$-invariant subspace $W$ other than $\{0\}$ and $V$ itself, then the action on $W$ itself is a representation of $G$. This leads to the following important definition of irreducible representations.
	
	\begin{definition}
		A representation $\rho$ of $G$ is irreducible if $\rho_V\neq \emptyset$ and $\rho_V$ has no $G$-invariant subspaces other than $\{0\}$ and $\rho_V$.
	\end{definition}
	
	We will denote the set of all irreducible representations of $G$ up to isomorphism by $\irr(G)$.

	\begin{fact}
		\label{fact:subgroup_rep}
		Let $G$ be a group and $H$ be any subgroup of $G$, if $\rho\in \irr(G)$  then $\rho$ restricted to $H$ is also a (not necessarily irreducible) representation of $H$.
	\end{fact}
	
	\begin{definition} The tensor product of two representations $\rho$ and $\tau$ of a group $G$  is
		the representation $\rho \otimes \tau$ on $\rho_V \otimes \tau_V$ defined by the condition
		$$(\rho\otimes \tau)(g)(v\otimes w) = \rho(g)(v)\otimes \tau(g)(w),$$
		and extended to all vectors in $\rho_V \otimes \tau_V$  by linearity.
	\end{definition}

	\begin{definition}
		The direct sum of two representations $\rho$ and $\tau$ is the space $\rho_{V} \oplus \tau_{V}$	with the block-diagonal action $\rho\oplus \tau$ of $G$.
	\end{definition} 
	
	If the representation is not irreducible, then by an appropriate change of basis $\rho$ can be converted into a block diagonal matrix with blocks corresponding to the invariant subspaces. Thus, any representation can be completely decomposed into a direct sum of irreducible representations of $G$, by applying an appropriate unitary transformation. Note that this decomposition is {\em unique}.  We use the following notation to denote the decomposition of a reducible representation: If $\rho$ is a reducible representation of $G$ then $\rho \cong \oplus_{i} n_i \rho_i$, where each $i$ we have {\em distinct} $\rho_i\in \irr(G)$ and $n_i$ denotes the multiplicity of $\rho_i$ in the decomposition. It will be convenient to think of this representation as a block diagonal matrix with $\rho_i$ as the blocks along the diagonal with multiplicity $n_i$.

	The following proposition shows that matrix entries of irreducible representations are 'orthogonal' with respect to a {\em symmetric bilinear form}, unless they are conjugates of each other -- in which case the corresponding product is the inverse of the dimension of the representation.
	\begin{proposition}
		\label{prop:orthmatrixentries}
		If $\rho$ and $\tau$ are two non-isomorphic irreducible representations of $G$ then for any $i,j,k,l$ we have
		\begin{equation}
			\label{eq:diffG_Product}
			\langle  (\rho)_{ij}\mid ( \tau)_{kl}\rangle_G = 0,
		\end{equation}
		where $\langle  f_1\mid  f_2\rangle_G := \frac{1}{|G|}\sum_{g\in G} f_1(g) f_2(g^{-1})$ (called a ``symmetric bilinear form''). Also,
		\begin{equation}
			\label{eq:sameG_Product}
			\langle  (\rho)_{ij}\mid ( \rho)_{kl}\rangle_G = \frac{\delta_{il}\delta_{jk}}{\dim(\rho)},
		\end{equation}
		where $\delta_{ij}$ is the delta-function which is $1$ if $i=j$ and $0$ otherwise. 
	\end{proposition}

	\subsubsection{Fourier analysis on non-abelian group}
	In this paper, we will be interested in studying $L^2(G)$, the space of functions from
	a finite group $G$ to the complex numbers $\C$. 
	
	\begin{definition}
		\label{def:norm}
		Define the inner product $\langle \cdot, \cdot\rangle_{L^2(G)}$ on $L^2(G)$ by 
		$$ \langle f, g \rangle_{L^2(G)} = \E_{x\in G} [f(x) \comp{g(x)}].$$
	\end{definition}

	We can define a character for every representation of a group.
	\begin{definition} 
		The character of a representation $\rho$ is the function $\chi_\rho : G \rightarrow \C$ defined by 
		$\chi_\rho(g) = \trace(\rho(g))$.
	\end{definition}

	The following proposition shows that the characters corresponding to the irreducible representations of a group are orthogonal to each other.
	
	\begin{proposition}[Orthogonality of characters]
		\label{prop:orthochar}
		For $\rho, \tau\in \irr(G)$, we have
		\begin{equation*}
			\frac{1}{|G|}\sum_{g\in G}\chi_{\rho}(g) \overline{\chi_{\tau}(g)}= \begin{cases}
				1\quad \rho_V \cong \tau_V,\\
				0 \quad otherwise.
			\end{cases}
		\end{equation*}
		
	\end{proposition}

	We use \Cref{prop:orthmatrixentries} many times in the proof. For convenience, we note an important identity that follows from  \Cref{prop:orthmatrixentries} (by setting $\tau$ to be the trivial map $\trivrep$).
	\begin{proposition}
		\label{prop:summapzero}
		If $\rho\in \irr(G) \hspace{-3pt}\setminus \hspace{-3pt}\trivrep$, $\sum_{g\in G} \rho(g) = 0$. 
	\end{proposition}
	
	We have the following proposition. It also shows that the maximum dimension of any irreducible representation of $G$ is at most $\sqrt{G}$.
	\begin{proposition}
		\label{prop:sumdimsquare}
		$$\sum\limits_{\rho \in \irr(G)} \dim(\rho)\chi_\rho(g) = \begin{cases}
			|G| \quad g= 1_G,\\
			0 \quad otherwise.
		\end{cases} $$
		This implies the following:
		$$\sum_{\rho\in \irr(G)} \dim(\rho)^2 = |G|.$$
	\end{proposition}

	\begin{definition}
		\label{def:conv}
		For two functions $f, g\in L^2(G)$ their {\em convolution} $f\conv g \in L^2(G)$ is defined as 
		$$(f\conv g)(x) := \E_{y\in G}[f(y)g(y^{-1}x)].$$
	\end{definition}
	
	For an abelian group, any function $f : G \rightarrow \C$ can be written as linear combinations of characters, i.e., the characters span the whole space $L^2(G)$. However, for non-abelian groups, characters form an orthonormal basis only for the set of {\em class functions} -- maps which are constant on {\em conjugacy classes}. A conjugacy class in $G$ is a nonempty subset $H$ of $G$ such that the following two conditions hold: 	Given any $x,y \in H$, there exists $g \in G$ such that $gxg^{-1} = y$, and if $x \in H$ and $g \in G$ then $gxg^{-1} \in H$. Since this is an equivalence class, any group is a collection of disjoint conjugacy classes.

	As in the Abelian case, we can understand operations like inner product, convolution, etc., using the Fourier transform, which is defined as follows:
	
	\begin{definition}
		For a function $f\in L^2(G)$, define the Fourier transform of $f$ to be the element $\hat{f} \in \prod_{\rho\in \irr(G)} \End{\rho_V}$ given  by  
		$$\hat{f}(\rho) = \E_{x\in G}[ f(x) \rho(x)] \in \End{\rho_V}.$$
	\end{definition}

	\begin{definition}
		Let $V$ be a finite-dimensional complex inner product space. Define an inner product  $\langle \cdot , \cdot\rangle_{\End{V}}$ on $\End{V}$ by
		$$ \langle A, B\rangle_{\End{V}} = \trace(A\ctranspose{B}).$$
	\end{definition}
	
	We can now state the Fourier inversion theorem.
	\begin{proposition}[Fourier inversion theorem]
		For $f\in L^2(G)$ we have
		$$f(x) = \sum_{\rho\in \irr(G)} \dim(\rho)\cdot \langle \hat{f}(\rho) , \rho(x)\rangle_{\End{\rho_V}}.$$
	\end{proposition}
	
	We have the following simple identities (See \cite{Terras99} for the proofs).
	\begin{proposition}[Plancherel's identity]
		\label{prop:plancherels}
		$$ \langle f, g \rangle_{L^2(G)} = \sum_{\rho\in \irr(G)} \dim(\rho)\cdot  \langle \hat{f}(\rho) , \hat{g}(\rho)\rangle_{\End{\rho_V}}.$$
	\end{proposition}

	\begin{proposition}[Parseval's identity]
		\label{prop:parsevals}
		$$ \E_{x\in G}[|f(x)|^2] = \sum_{\rho \in \irr(G)} \dim(\rho)\cdot \hsnorm{\hat{f}(\rho)}^2,$$
		where $\hsnorm{A}:= \sqrt{\langle A, A\rangle_{\End{V}}} = \sqrt{\trace(A\ctranspose{A})}= \sqrt{\sum_{ij} |A_{ij}|^2}$.
	\end{proposition}
	
	Note that the norm $\hsnorm{\cdot}$ satisfies a triangle inequality.
	\begin{claim}
		\label{claim:normteq}
		$\hsnorm{AB}\leq \hsnorm{A}\cdot \hsnorm{B}.$
	\end{claim}
	\begin{proof}
		$
		\hsnorm{AB}^2 = \sum_{ij} |(AB)_{ij}|^2 \leq \sum_{ij} \left(\sum_{k} |A_{ik} B_{kj}|\right)^2 
		$.
		Using the Cauchy-Schwarz inequality on the inner sum,
		\begin{align*}
			\hsnorm{AB}^2 	\leq  \sum_{ij} \left(\sum_{k} |A_{ik}|^2 \right)\left(\sum_{l} |B_{l j}|^2\right) = \sum_{ijkl} |A_{ik}|^2 |B_{l j}|^2 =  \left(\sum_{ik} |A_{ik}|^2 \right) \left(\sum_{l j} |B_{l j}|^2 \right) = \hsnorm{A}^2 \cdot \hsnorm{B}^2.
		\end{align*}
	\end{proof}

	% \begin{claim}
	% 	\label{claim:Uni_norm}
	% 	Let $A$ be any matrix and $U$ be any unitary matrix, then $\hsnorm{UA} = \hsnorm{A}$.
	% \end{claim}
	% \begin{proof}
	% 	Let $V$ be a unitary matrix which converts $U$ to the identity matrix, i.e., $VUV^\star = I$. Since the change of basis does not change the $\hsnorm{\cdot}$, we have
	% 	$$\hsnorm{UA} = \hsnorm{VUAV^\star} = \hsnorm{VUV^\star VAV^\star} = \hsnorm{I VAV^\star} = \hsnorm{A}.$$
	% \end{proof}
	
	\begin{proposition}[Convolution theorem]
		For $f, g\in L^2(G)$ we have 
		$$\widehat{f\conv g}(\rho) = \hat{f}(\rho) \hat{g}(\rho).$$
	\end{proposition}

	\subsection{Important claims}
	In this section, we prove a few statements that will be used in the soundness analysis. The following claim shows that the character functions always come in `pairs' with respect to the complex conjugation.
	\begin{claim}
		\label{claim:dim1conj}
		Let $G$ be any non-abelian group. For every $\rho \in \irr(G)$, such that $\dim(\rho)=1$, there exists $\widetilde{\rho}\in \irr(G)$ with $\dim(\widetilde{\rho}) = 1$ such that 
		$$ \chi_\rho(g) = \overline{\chi_{\widetilde{\rho}}(g)}, \quad\quad \forall g\in G.$$
	\end{claim}
	\begin{proof}
		We claim that the set of characters corresponding to dimension $1$ irreducible representations of $G$ forms a group under point-wise multiplication. This will be enough to show the claim.

		Let $G' = G/\commutator{G}{G}$ be the abelian quotient group. Assume $\rho$ is a degree $1$ representation of $G$. Then it satisfies $\rho(a)\rho(b) = \rho(ab)$ for all $a, b\in G$. Define a map $\Gamma_\rho : G' \rightarrow \C$ as $\Gamma_\rho(g') = \rho(g)$ where $g' = g\commutator{G}{G}$. This is a well-defined map as
		$$\rho(aba^{-1}b^{-1}) = \rho(a)\rho(b) \rho(a^{-1})\rho(b^{-1}) = \rho(a)\rho(a^{-1})\rho(b)\rho(b^{-1}) = 1.$$ 
		Thus, the map $\rho$ is constant on every coset of $\commutator{G}{G}$ and hence $\Gamma_\rho$ is well defined. The set of all $\{\Gamma_\rho \mid \rho \in \irr(G), \dim(\rho) = 1\}$ is the set of all the multiplicative characters of the abelian group $G'$ and hence forms a group under coordinate-wise multiplication. There is a one-to-one correspondence between the coordinate-wise multiplicative action of $\Gamma_\rho$'s and $\rho$'s. Thus, $\{\chi_\rho \mid \rho \in \irr(G), \dim(\rho) = 1\}$ form a group under point-wise multiplication.
	\end{proof}

	\subsection{Functions on $G^n$}
	
	For any non-abelian group $G$ and $n\geq 1$, we have a group $G^n$ where the group operation is defined coordinate-wise.
	The irreducible representations of $G^n$ are precisely those representations obtained by taking tensor products of $n$ irreducible representations of $G$.
	\begin{proposition}[\cite{Terras99}]
		\label{prop:irr_Gn}
		The set of irreducible representations of $G^n$ is given by 
		$$\irr(G^n) = \{ \alpha \mid  \alpha = \otimes_{i\in [n]} \rho_i \mbox{ where } \rho_i \in \irr(G)\}.$$
	\end{proposition}
	We denote $\alpha$ by the corresponding tuple $(\rho_1, \rho_2, \ldots, \rho_n)$. We define the weight of a representation $\alpha = (\rho_1, \rho_2, \ldots, \rho_n)$ (denoted by $|\alpha|$) to be the number of non-trivial representations in $(\rho_1, \rho_2, \ldots, \rho_n)$.

	We will be working with functions $f: G^n \rightarrow G$ which are {\em folded}. $f$ is said to be folded if $f(c\V x) = cf(\V x)$ for all $c\in G$ and $\V x \in G^n$. The following claim shows that for all functions $g(\V x):=\rho(f(\V x))_{ij}$ where $\dim(\rho)\geq 2$ and $1\leq i,j\leq \dim(\rho)$,  all the Fourier coefficients corresponding to representations of dimension $1$ are zero, if $f$ is folded.
	
	\begin{lemma}[Lemma 2.25 in \cite{BhangaleK21}]
		\label{lemma:folded_zero_coeff}
		Let $f:G^n \rightarrow G$ be any folded function and $g(\V x):=\rho(f(\V x))_{ij}$ where $\rho \in \irr(G), \dim(\rho)\geq 2$ and $1\leq i,j\leq \dim(\rho)$. Let $\alpha$ be {\em any} representation of $G^n$ such that $\dim(\alpha) = 1$, then $\hat{g}(\alpha) = 0$.
	\end{lemma}

	Fix any surjective projection map $\pi: [L]\rightarrow [R]$ for some $L\geq R$. Consider the following subgroup of $G^L$ given by the elements 
	$$ \{ (x\circ \pi) \in G^L \mid x\in  G^R\},$$
	where $(x\circ \pi)_i = x_{\pi(i)}$. Let us denote this group by $\pi(G^L)$. Note that this group is isomorphic to $G^R$. Thus, any representation $\alpha \in \irr(G^L)$ (which is a representation of $G^R$ using \Cref{fact:subgroup_rep}), can be decomposed into irreducible representations of $G^R$.

	\subsection{Notations}
	
	Whenever possible, we use the notation $\alpha, \beta$ to denote the representations of a group $G^n$ and $\rho, \tau$ for group $G$. Also, we use bold letters $\V x, \V c$ to denote the elements of $G^n$.
	
	For a representation $\alpha\in \irr(G^n)$ where $\alpha = \otimes_{i=1}^n \rho_i$, we use the notation $\dimgeqi{\alpha}{k}$ to denote the number of $i\in [n]$ such that $\dim(\rho_i)\geq k$.

\section{An Approximation Algorithm}
    \label{sec:approx_alg}
	In this section, we give an approximation algorithm for $\text{Max-E$3$-LIN}_{S}(G)$ for any group $(G, \gp)$ and $S\subseteq G$. This algorithm is a straightforward generalization of the algorithm for abelian groups $G$ discussed in the introduction. We repeat it here for completeness. 
	
	\begin{theorem}
		\label{thm:approx_alg}
		There exists a $\frac{|S|}{|H_S|}$-approximation algorithm for $\text{Max-E$3$-LIN}_{S}(G)$, where $H_S$ is the smallest normal subgroup such that (i) $[G, G] \subseteq H_S$, and (ii) $S$ is a subset of some coset of $H_S$. 
	\end{theorem}
	\begin{proof}
		Let $\Phi$ be an instance of $\text{Max-E$3$-LIN}_{S}(G)$ with constraints $(C_1, C_2, \ldots, C_m)$ over the variables $X = \{ x_1, x_2, \ldots, x_n\}$. We first convert the set of constraints to a system of equations, denoted by $\tilde{\Phi}$, over the group $(Q,+):= G/H_S$ with variables $Y = \{ y_1, y_2, \ldots, y_n\}$. Note that $H_S$ is a normal subgroup of $G$ containing the commutator subgroup such that $S\subseteq gH_S$ for some $g\in G$. Thus, $Q$ is an abelian group. 
		
		Consider a constraint $C_i$ which is of the form $(a_{i_1}\gp x_{i_1})\gp (a_{i_2}\gp x_{i_2})\gp (a_{i_3}\gp x_{i_3})\in S$. We convert this to the equation over $Q$ as 
		$$[a_{i_1}]_Q + y_{i_1} + [a_{i_2}]_Q + y_{i_2} + [a_{i_3}]_Q + y_{i_3} = [S]_Q,$$
		where $[S]_Q$ is an element of $Q$ that corresponds to the coset of $H_S$ containing $S$, and the element $[g]_Q$ corresponds to the coset of $H_S$ containing $g$.
		
		As $\Phi$ is satisfiable, consider the satisfying assignment $\V \alpha: X\rightarrow G$ to $\Phi$. Consider the assignment $\tilde{\V \alpha} : Y\rightarrow Q$ given by the natural map $\tilde{\V \alpha}(y_i) = [\V \alpha(x_i)]_Q$ . It is easy to see that $\tilde{\V \alpha}$ satisfies all the equations from the instance $\tilde{\Phi}$, and hence, $\tilde{\Phi}$ is satisfiable.
		
		Since $\tilde{\Phi}$ is a system of equations over an abelian group $(Q,+)$, we can find a satisfying assignment to $\tilde{\Phi}$ in polynomial time using Gaussian elimination~\cite{RG99}. Let $\tilde{\V \beta}$ be the assignment returned by this procedure. To construct the final assignment to the $X$ variables, we simply set $x_i$ to be a random element from the coset $\tilde{\V \beta}(y_i)$. Let $\V \beta: X\rightarrow G$ be the random assignment given by the above procedure. It can be easily observed that $\V \beta$ satisfies a given constraint $C_i$ in $\Phi$ with probability $\frac{|S|}{|H_S|}$ and hence $\V \beta$ satisfies $\frac{|S|}{|H_S|}$ fraction of the constraints in expectation. The randomized algorithm can be easily derandomized using the method of conditional expectations.
	\end{proof}

\section{Hardness of $\text{Max-E$3$-LIN}_{S}(G)$}
    \label{sec:nonabelian_red}
	We start with some basic facts. For a nonabelian group $G$, the quotient group $G/[G,G]$ is an abelian group. The dual of $G/[G,G]$ is isomorphic to
	$$\{\chi_{\rho}~|~\rho\in\irr(G), \dim(\rho)=1\}$$
	We denote this subgroup of $\hat{G}$ as $\widehat{G/[G,G]}$.
	
	Similarly, consider any normal subgroup $H\trianglelefteq G$ such that $[G,G]\subseteq H$, the quotient group $G/H$ is an abelian group that is isomorphic to a subgroup of $G/[G,G]$. Furthermore, the dual of $G/H$ is isomorphic to 
	$$\{\chi_{\rho}~|~\rho\in\irr(G),\dim(\rho)=1, \chi_\rho(h)=1, \forall h\in H\}.$$
	
    \paragraph{Main Reduction.} We now give a reduction from a \LLC\ instance, denoted by,  $\mathcal{H}=(\mathcal{V}=\{V_1,\dots,V_T\}, \{\Pi_{i,j}\}_{1\leq t < t'\leq T}\}, \{[R_t]_{t\in [T]}\}$ to a $\text{Max-E$3$-LIN}_{S}(G)$ instance $\Phi$ over a non-abelian group $G$. For $\delta>0$, we will use the following setting of $T$ and $r$ in Theorem~\ref{thm:layered_lc_hardness}
 $$ 2^{-r}\leq \min\left\{ \frac{\delta^{10}}{(2|G|)^{20}}, \frac{\delta^2}{10|G|^{10K}}\right\}, \quad T \geq \left(\frac{8|G|^3}{\delta}\right)^4,$$
where $K := \frac{8|G|^6}{\delta^2}$.

	Consider a \LLC\ instance $\mathcal{H}=(\mathcal{V}=\{V_1,\dots,V_T\}, \{\Pi_{t,t'}\}_{1\leq t < t' \leq T}\}, \{[R_t]_{t\in [T]}\}$. For all $t\in [T]$ and for each $v\in V_t$, we create a cluster $C[v]$ of literals of size $|G|^{R_t}$. In each cluster $C[v]$, each literal is indexed by a string of length $R_t$. For any string $(1_G, \V y)\in G^{R_t}$, its corresponding literals are $g\gp (1_G,\V y)$ for $g\in G$, where the string $g\gp (1_G,\V y)$ is $(g, g\gp y_1,\ldots, g\gp y_{R_{t-1}})$. 
	
	An assignment to the instance that we are going to create is given by the maps $A_v: G^{R_t} \rightarrow G$ for all $v\in V_{t}$ and all $t\in[T]$. Note that any such assignment is assumed to be folded, i.e., $A_v(g\gp (1_G, \V y)) = g\gp A_v((1_G,\V y))$.
	
	The distribution on the constraint of the reduced instance $\Phi$ of $\text{Max-E$3$-LIN}_{S}(G)$ is given by the following PCP verifier.
	\begin{enumerate}
        \item Pick a uniformly random pair $(t,t')$ satisfying $1\leq t<t'\leq T$.
		\item Choose an edge constraint $\pi_{uv}:[R_t] \rightarrow [R_{t'}]$ from $\mathcal{H}$ uniformly at random.
		\item Sample a string $\V x\sim G^{R_{t'}}$ and $\V y\sim G^{R_t}$ independently and uniformly at random.
		\item Sample an element $\V s \in S^{R_t}$ uniformly at random.
		\item For each $j\in [R_t]$, set $z_j=y_j^{-1}\gp x_{\pi_{uv}(j)}^{-1}\gp s_j$.
		\item Accept if and only if $A_v(\V x)\gp A_u(\V y)\gp A_u(\V z)\in S$.
	\end{enumerate}
	
	\subsection{Completeness}

	If $\mathcal{H}$ is fully satisfiable, then there exists a corresponding assignment $\sigma$ such that all the constraints are satisfiable.  Let $A_v(\V x)= x_{\sigma(v)}$ and $A_u(\V y)= y_{\sigma(u)}$, i.e., the dictator functions. Then, the test passes as,
	\begin{align*}
		&A_v(\V x) \gp A_u(\V y) \gp A_u(\V z)\\
		&=x_{\sigma(v)}\gp y_{\sigma(u)} \gp z_{\sigma(u)}\\
		&=x_{\sigma(v)}\gp y_{\sigma(u)} \gp (y_{\sigma(u)})^{-1} \gp (x_{\pi_{u,v}(\sigma(u))})^{-1}\gp s_{\sigma(u)} \\
		&=x_{\sigma(v)}\gp y_{\sigma(u)} \gp (y_{\sigma(u)})^{-1} \gp (x_{\sigma(v)})^{-1}\gp s_{\sigma(u)} \tag{Using $\pi_{u,v}(\sigma(u)) = \sigma(v)$}\\
		&=s_{\sigma(u)}\in S.
	\end{align*}
	Hence, the test always passes. Thus, the value of the instance $\Phi$ is $1$.
	
	\subsection{Soundness}
	In this section, we prove the soundness of the analysis. 
	
	\begin{lemma}
    \label{lemma:main_soundness}
         % If the value of the $\text{Max-E$3$-LIN}_{S}(G)$ instance $\Phi$ is greater than $\frac{|S|}{|H_S|}+ \delta$, then for constants $C=O_G(\delta^{-4/d_0})$ and $K=\Omega_{G}(\delta^{-4/d_0})$ independent of $R_t$, there exists a labeling to $\mathcal{H}$ with value of at least $\min(\frac{\delta^2}{16C|G|^{C/2}}, \frac{\delta^2}{40|S|^2|G|10^K})$.\abnote{There exists a labeling to two layers that satisfies at least xx fraction of the constraints between the layers. For every $\delta>0$, if $\mathcal{H}$ is at most $\min(\frac{\delta^2}{16C|G|^{C/2}}, \frac{\delta^2}{40|S|^2|G|10^K})$ satisfiable, where $C=O_G(\delta^{-4/d_0})$ and $K=\Omega_{G}(\delta^{-4/d_0})$, then the $\text{Max-E$3$-LIN}_{S}(G)$ instance $\Phi$ has value at most $\frac{|S|}{|H_S|}+ \delta$, }
         For every $\delta>0$, if the \LLC\ instance $\mathcal{H}$ is at most $2^{-r}$, then the $\text{Max-E$3$-LIN}_{S}(G)$ instance $\Phi$ is at most $\frac{|S|}{|H_S|}+ \delta$ satisfiable.
	\end{lemma}
	\begin{proof}
		Fix the assignment $\{A_v\}_{v\in \mathcal{V}}$ to the instance $\Phi$. We define the value of an assignment $A$, $\textrm{value}(A)$, as the probability that the above test passes. The following expression gives the value of this assignment $A$,  
		
		\begin{align*}
			\mathrm{value}(A)  = \E_{1\leq t<t'\leq T}\left[\E_{\substack{\pi_{uv}\in \Pi_{t,t'}\\ (\V x, \V y, \V s)}} \left[ \sum_{s\in S} \left[\frac{1}{|G|} \sum_{\rho\in \irr(G)}\dim(\rho)\cdot \chi_{\rho}(A_v(\V x)\gp A_{u}(\V y)\gp  A_{u}(\V z)\gp s^{-1}) \right] \right]\right].
		\end{align*}
		By Proposition~\ref{prop:sumdimsquare}, this expression equals $1$ if and only if $A_v(\V x)\gp A_{u}(\V y)\gp  A_{u}(\V z)\in S$, and $0$ otherwise. We can rewrite this expression according to the representations $\rho\in\irr(G)$,
	\begin{align}
		\mathrm{value}(A) &= \frac{1}{|G|}\E_{1\leq t<t'\leq T}\left[\sum_{\rho\in\irr(G)}\sum_{s\in S}\E_{\substack{\pi_{uv}\in \Pi_{t,t'}\\ \V x ,\V y, \V s}}\left[ \dim(\rho)\cdot 
		\chi_{\rho}(A_v(\V x)\gp A_{u}(\V y)\gp  A_{u}(\V z)\gp s^{-1}) \right]\right]\notag\\
		&=\frac{1}{|G|}\E_{1\leq t<t'\leq T}\left[\sum_{\substack{\rho\in\widehat{G/H_S}\\\dim(\rho)=1}}\sum_{s\in S}\E_{\substack{\pi_{uv}\in \Pi_{t,t'}\\\V x ,\V y, \V s}}\left[\dim(\rho)\chi_{\rho}(A_v(\V x)\gp A_{u}(\V y)\gp  A_{u}(\V z)\gp s^{-1}) \right]\right]\label{eq:constant}\\
		&+ \frac{1}{|G|}\E_{1\leq t<t'\leq T}\left[\sum_{\substack{\rho\notin\widehat{G/H_S}\\\dim(\rho)=1}}\sum_{s\in S}\E_{\substack{\pi_{uv}\in \Pi_{t,t'}\\\V x ,\V y, \V s}}\left[\dim(\rho)\chi_{\rho}(A_v(\V x)\gp A_{u}(\V y)\gp  A_{u}(\V z)\gp s^{-1}) \right]\right]\label{eq:nonconstant}\\
		&+\frac{1}{|G|}\E_{1\leq t<t'\leq T}\left[\sum_{\dim(\rho)\geq 2}\sum_{s\in S}\E_{\substack{\pi_{uv}\in \Pi_{t,t'}\\\V x ,\V y, \V s}}\left[ \dim(\rho)\cdot \chi_{\rho}(A_v(\V x)\gp A_{u}(\V y)\gp  A_{u}(\V z)\gp s^{-1}) \right]\right].\label{eq:greaterdim}
	\end{align}
	Term~(\ref{eq:constant}) is a constant between any two layers. As for term~(\ref{eq:nonconstant}), we prove that they can be used to decode a valid assignment to any pair of layers unless they are negligible along a random path $p$. Finally, for term~(\ref{eq:greaterdim}), we use~\cite{BhangaleK21new} as a black box and show they are negligible along a random path $p$.
    
    For term~\ref{eq:constant} and term~\ref{eq:nonconstant}, since they have $\dim(\rho)=1$, the character $\chi_\rho$ is a homomorphism. Thus, we have
	\begin{align*}
		\dim(\rho)\chi_{\rho}(A_v(\V x)\gp A_{u}(\V y)\gp  A_{u}(\V z)\gp s^{-1})=\chi_{\rho}(A_v(\V x)\gp A_{u}(\V y)\gp  A_{u}(\V z))\cdot \chi_{\rho}(s^{-1}).
	\end{align*}
	For term~(\ref{eq:constant}), as $\chi_\rho$ is a 1-bounded function, it's upper bounded by
	\begin{align*}
		\frac{1}{|G|}\sum_{\substack{\rho\in\widehat{G/H_S}\\\dim(\rho)=1}}\sum_{s\in S}1=\frac{1}{|G|}|S|\left|\frac{G}{H_S}\right|=\frac{|S|}{|H_S|}.
	\end{align*}

    	\subsubsection{Bounding expressions in (\ref{eq:nonconstant})}
	We now bound the term~(\ref{eq:nonconstant}). Using $|\chi_\rho(s^{-1})|\leq1$, we have (\ref{eq:nonconstant}) is at most
	\begin{align*}
		\frac{|S|}{|G|}\sum_{\substack{\rho\notin\widehat{G/H_S}\\\dim(\rho)=1}}\E_{\substack{(u,v)\\\V x ,\V y, \V s}}\left[\chi_{\rho}(A_v(\V x))\cdot \chi_{\rho}(A_{u}(\V y))\cdot  \chi_\rho(A_{u}(\V z)) \right],
	\end{align*}
	We argue that if this expression is large, then a decoding strategy exists for the Label Cover instance.
	We need the following simple lemma.
	\begin{lemma}\label{lemma:folding}
		Let $h: G^n \rightarrow G$ be any folded function, $\beta$ is a representation of $G^n$. Define $s(\beta):=\{i~|~\beta_i\notin \widehat{G/H_S}\}$. Let $g(\V x) = \chi_{\rho}(h(\V x))$ for $\chi_\rho\notin \widehat{G/H_S}$ and $\dim(\rho)=1$, then $\hat{g}(\beta) = 0$ for all $\beta$ with $s(\beta)=\emptyset$.
	\end{lemma}
	\begin{proof}
		By definition, we have 
		\begin{align*}
			\hat{g}(\beta)&=\E_{\V x \in G^n} [\chi_{\rho}(h(\V x))\beta(\V x)]\\
			&=\E_{\substack{\V y\in G^{n-1}\\x_1=1_G}}\left[\E_{c\in G}[\chi_\rho(h(c\gp (1_G, \V y)))\beta(\V c\gp(1_G,\V y))]\right]\\
			&=\E_{\V y\in G^{n-1}}\left[\E_{c\in G}[\chi_\rho(c)\cdot \chi_\rho(h((1_G, \V y)))\cdot \beta(\V c) \cdot \beta((1_G,\V y))]\right]\\
			&=\E_{\V y\in G^{n-1}}\chi_{\rho}(h((1_G, \V y)))\beta((1_G, \V y))\cdot \E_{c\in G}\left[\chi_\rho(c)\beta(\V c)\right].
		\end{align*}
		If $\beta$ satisfies $|s(\beta)|=0$, then $\beta(\V c)=\bigotimes_{i=1}^{L}\beta_i(c)$ is a complex number. Therefore, there always exists a $\rho'$ such that $\dim(\rho')=1$ and $\chi_\rho(c)\beta(\V c)=\chi_{\rho'}(c)$.  Furthermore, such $\rho'\notin\widehat{G/H_S}$ since for any element $q\in G/H_S$, $\beta(q)=1$ and $\chi_\rho(q)\neq 1$, indicating that $\chi_{\rho'}(q)\neq 1$. Hence, by Proposition~\ref{prop:summapzero},
		\begin{align*}
			\hat{g}(\beta)=\E_{\V y\in G^{n-1}}\chi_{\rho}(h((1_G, \V y)))\beta((1_G, \V y))\cdot \E_{c\in G}[\chi_{\rho'}(c)]=0.
		\end{align*}
	\end{proof}
	
    We now prove the following main lemma from this section.
	\begin{lemma}
		\label{lemma:soundness_dimone}
		If the \LLC\ instance $\mathcal{H}$ is at most $2^{-r}$ satisfiable, then for any $\dim(\rho)=1$ such that $\rho\notin\widehat{G/H_S}$, 
		\begin{align}
			\label{eq:soundness_decode}
			\left|\E_{1\leq t < t'\leq T}\left[ \E_{\pi_{uv}\in \Pi_{t,t'}}\left[\E_{\V x, \V y} \left[ \chi_{\rho}(A_v(\V x))\cdot \chi_{\rho}(A_{u}(\V y))\cdot  \chi_\rho(A_{u}(\V z)) \right] \right]\right] \right|\leq \frac{\delta}{2|G|},
		\end{align}
	\end{lemma}

	\begin{proof}
		Consider two layers $U$ and $V$ whose alphabets are $[L]$ and $[R]$ and an edge constraint $e=(u,v)$ such that $u\in U$ and $v\in V$. Let $\pi$ denote the projection constraint on $e$. Let $f_v(\V x)=\chi_{\rho}(A_v(\V x))$, $g_u(\V x)=\chi_{\rho}(A_u(\V x))$ and $h_u^{\V s}(\V x)=\chi_\rho(A_u(\V x\gp \V s))$. With these notations, we have
        $$ \chi_{\rho}(A_{u}(\V y))\cdot  \chi_\rho(A_{u}(\V z)) = g_u(\V y) \cdot h_u^{\V s}(\V y^{-1} \gp (\V x\circ \pi)^{-1}) = (g_u\conv h_u^{\V s}) ((\V x\circ \pi)^{-1}),$$
        where $(\V x\circ \pi)_{j} = x_{\pi(j)}$. Thus, the inner expectation can be written as
		\begin{align*}
			&\E_{\V x, \V s}[f_v(\V x)\cdot (g_u \conv h^{\V s}_u)((\V x \circ \pi)^{-1})]\\
            &=\E_{\V x, \V s}\left[\sum_{\alpha\in\irr(G^R)}\dim(\alpha)\trace(\hat{f_v}(\alpha)\cdot \alpha(\V x^{-1}))\sum_{\beta\in \irr(G^L)}\dim(\beta)\trace(\hat{g_u}(\beta)\hat{h_u^{\V s}}(\beta)\cdot\beta(\V x \circ \pi))\right]\\
			&=\E_{\V x, \V s}\left[\sum_{\alpha,\beta}\dim(\alpha)\dim(\beta)\trace(\hat{f_v}(\alpha)\cdot \alpha(\V x^{-1}))\trace(\hat{g_u}(\beta)\hat{h_u^{\V s}}(\beta)\cdot\beta(\V x \circ \pi))\right]\\
			&=\sum_{\alpha,\beta}\dim(\alpha)\dim(\beta)\E_{\V x, \V s}\left[\trace(\hat{f_v}(\alpha)\cdot \alpha(\V x^{-1}))\trace(\hat{g_u}(\beta)\hat{h_u^{\V s}}(\beta)\cdot\beta(\V x \circ \pi))\right].
		\end{align*}
		Denote 
		$$\termm{e}{\alpha}{\beta}:=\dim(\alpha)\dim(\beta)\E_{\V x, \V s}[\trace(\hat{f_v}(\alpha)\cdot \alpha(\V x^{-1}))\trace(\hat{g_u}(\beta)\hat{h_u^{\V s}}(\beta)\cdot\beta(\V x \circ \pi))],$$
		we have
		\begin{align*}
			\termm{e}{\alpha}{\beta}&=\dim(\alpha)\dim(\beta)\E_{\V x, \V s}\left[\sum_{1\leq p,q\leq \dim(\alpha)}\hat{f_v}(\alpha)_{pq}\cdot \alpha(\V x^{-1})_{qp}\sum_{1\leq i,k\leq \dim(\beta)}\hat{g_u}(\beta)\hat{h_u^{\V s}}(\beta)_{ik}\cdot\beta(\V x \circ \pi)_{ki}\right]\\
			&=\dim(\alpha)\dim(\beta)  \E_{\V s}\left[\sum_{p,q,i,k}\hat{f}_v(\alpha)_{pq}\cdot (\hat{g}_u(\beta)\hat{h^{\V s}_u}(\beta))_{ik}\cdot  \E_{\V x}\left[\alpha(\V x^{-1})_{qp}\cdot\beta(\V x\circ \pi)_{ki}\right]\right],
		\end{align*}
		where $(i,k)$ are the tuples $i=(i_1,i_2,\dots,i_L)$ and $k=(k_1,k_2,\dots,k_L)$. Similarly, $(p,q)$ are tuples $p=(p_1,p_2,\dots,p_R)$ and $q=(q_1,q_2,\dots,q_R)$. Then,
		\begin{align}
			\E_{\V x}\left[\alpha(\V x^{-1})_{qp}\cdot\beta(\V x\circ \pi)_{ki}\right]&=\E_{\V x}\left[\prod_{l=1}^R\alpha_l(x_l^{-1})_{q_l p_l}\prod_{l'=1}^{L}\beta_{l'}(x_{l'})_{k_{l'}i_{l'}}\right]\\
			&=\prod_{l=1}^R\E_{\V x}\left[\alpha_l(x_l^{-1})_{q_l p_l}\prod_{l'\in\pi^{-1}(l)}\beta_{l'}(x_{l'})_{k_{l'}i_{l'}}\right].\label{eq:alphabetaconsistent}
		\end{align}
            We argue that $s(\alpha)\subseteq \pi(s(\beta))$ to make this expression non-zero. Suppose that there exists a $\ell$ such that for all $\ell'\in\pi^{-1}(\ell)$, $\dim(\beta_{\ell'})=1$, then the product of all such $\beta_{\ell'}$ must be 1 dimensional. According to Proposition~\ref{prop:orthmatrixentries}, the expectation is $0$ unless $\alpha_\ell$ is also 1 dimensional and isomorphic to the product of $\beta_{\ell'}$. However, for an $\alpha_{l}\notin\widehat{G/H_S}$, if $\beta_{\ell'}\in\widehat{G/H_S}$ for all $\ell'$, then the product of $\beta_{\ell'}$ belongs to $\widehat{G/H_S}$ and is not isomorphic to $\alpha_\ell$. Therefore, for the expectation to be non-zero, there exists some $\ell'$ such that $\beta_{\ell'}\notin\widehat{G/H_S}$, indicating $s(\alpha)\subseteq s(\beta)$.
            
		Define a function $F_\beta^{ki}(\V x^{-1}):=\beta(\V x\circ \pi)_{ki}$, and note that
		\begin{equation}
			\label{eq:sum_one}
			\sum_{i} \|F_\beta^{ki}\|_2^2 = \sum_{i}  \E_{\V x}[ |{\beta}(\V x^{-1}\circ \pi)_{ki}|^2] = \E_{\V x} \left[\sum_{i}  |{\beta}(\V b^{-1}\circ \pi)_{ki}|^2\right] = 1,
		\end{equation}
		where the last equality follows from the fact that the sum expression is exactly the norm of the $k$-th row of the representation $\beta$. Since $\beta(.)$ is unitary, the norm of its row is always $1$.
		Using the function $F_\beta^{ki}$, we further simplify the expectation $  \E_{\V x}\left[\alpha(\V x^{-1})_{qp}\cdot\beta(\V x\circ \pi)_{ki}\right]$ as follows. 
		\begin{align*}
			\E_{\V x}\left[\alpha(\V x^{-1})_{qp}\cdot\beta(\V x\circ \pi)_{ki}\right]&=\E_{\V x}[\alpha(\V x^{-1})_{qp}F_{\beta}^{ki}(\V x^{-1})]\\&=\E_{\V x}\left[\alpha(\V x^{-1})_{qp}\sum_{\gamma}\dim(\gamma)\trace(\hat{F_{\beta}^{ki}}(\gamma)\gamma(\V x))\right]\\
			&=\E_{\V x}\left[\sum_{\gamma}\dim(\gamma)\sum_{p',q'}\hat{F_{\beta}^{ki}}(\gamma)_{p'q'}\gamma(\V x)_{q'p'}\alpha(\V x^{-1})_{qp}\right]\\
			&=\sum_{\gamma}\dim(\gamma)\sum_{p',q'}\hat{F_{\beta}^{ki}}(\gamma)_{p'q'}\E_{\V x}\left[\gamma(\V x)_{q'p'}\alpha(\V x^{-1})_{qp}\right].
		\end{align*}
		By Proposition~\ref{prop:orthmatrixentries}, the expectation is $0$ unless $\alpha=\gamma$,  $p'=q$ and $q'=p$. Hence, we have
		$$\E_{\V x}\left[\alpha(\V x^{-1})_{qp}\cdot\beta(\V x\circ \pi)_{ki}\right]=\hat{F_{\beta}^{ki}}(\alpha)_{pq}.$$
		Thus, we can express $\termm{e}{\alpha}{\beta}$ as,
		$$\termm{e}{\alpha}{\beta}=\dim(\alpha)\dim(\beta)\sum_{p,q,i,k}\hat{f}_v(\alpha)_{pq}\cdot \E_{\V s}\left[(\hat{g}_u(\beta)\hat{h^{\V s}_u}(\beta))_{ik}\right]\cdot \hat{F_{\beta}^{ki}}(\alpha)_{pq}.$$
		In addition, Lemma~\ref{lemma:folding} indicates that $s(\alpha)$ and $s(\beta)$ are non-empty. Therefore,
		\begin{align*}
			\E_{\V x, \V s}[f_v(\V x)\cdot (g \conv h^{\V s})((\V x \circ \pi)^{-1})]&=\sum_{\substack{\alpha,\beta\\|s(\alpha)|,|s(\beta)|\neq0\\s(\alpha)\subseteq \pi(s(\beta))}}\termm{e}{\alpha}{\beta}\\
			&=\underbrace{\sum_{\substack{\alpha,\beta\\|s(\alpha)|,|s(\beta)|\neq 0\\|s(\beta)|<C\\s(\alpha)\subseteq \pi(s(\beta))}}\termm{e}{\alpha}{\beta}}_{\Theta^{e(u,v)}_{\mathrm{Low}}}+\underbrace{\sum_{\substack{\alpha,\beta\\|s(\alpha)|,|s(\beta)|\neq 0\\|s(\beta)|\geq C\\s(\alpha)\subseteq \pi(s(\beta))}}\termm{e}{\alpha}{\beta}}_{\Theta^{e(u,v)}_{\mathrm{High}}},
		\end{align*}
		where in the last expression we break the summation based on $|s(\beta)|$.
		
		If~(\ref{eq:soundness_decode}) is not true, then we have
		$$\left|\E_{e(u,v)} \left[\Theta^{e(u,v)}_{\mathrm{Low}} + \Theta^{e(u,v)}_{\mathrm{High}} \right]\right| \geq \delta'.$$
		where $\delta' = \frac{\delta}{2|G|}$. We  will later show that that $\left|\E_{e(u,v)} \left[ \Theta^{e(u,v)}_{\mathrm{High}} \right]\right| \leq \delta'/2$. Assuming this, we have
		$$\left|\E_{e(u,v)} \left[ \Theta^{e(u,v)}_{\mathrm{Low}} \right]\right| \geq \delta'/2.$$
		We now show how to come up with a decoding strategy based on the above lower bound.

		\paragraph{Bounding the $\Theta^{e(u,v)}_{\mathrm{Low}}$ term.}Next, we argue that if $|\E_{e}[\Theta^{e(u,v)}_{\mathrm{Low}}]|$ is large, then we can decode an assignment to the \LC\ instance $\mathcal{H}$. We first simplify the expression,
		\begin{align*}
			|\Theta^{e(u,v)}_{\mathrm{Low}}|^2&=\left|\E_{\V s}\left[ \sum_{\alpha,\beta}\dim(\alpha)\dim(\beta)\sum_{p,q,i,k}\hat{f}_v(\alpha)_{pq}\cdot (\hat{g}_u(\beta)\hat{h^{\V s}_u}(\beta))_{ik}\cdot \hat{F_{\beta}^{ki}}(\alpha)_{pq}\right]\right|^2\\
			&=\left|\E_{\V s}\left[\sum_{\alpha,\beta}\dim(\alpha)\dim(\beta)\sum_{p,q,i,j,k}\hat{f}_v(\alpha)_{pq}\cdot \hat{g}_u(\beta)_{ij}\cdot\hat{h^{\V s}_u}(\beta)_{jk}\cdot \hat{F_{\beta}^{ki}}(\alpha)_{pq}\right]\right|^2\\
			&\leq \left(\sum_{\alpha,\beta}\dim(\alpha)\dim(\beta)\sum_{\substack{p,q\\i,j,k}}|\hat{f}_v(\alpha)_{pq}|^2|\hat{g}_u(\beta)_{ij}|^2\right)\left(\sum_{\alpha,\beta}\dim(\alpha)\dim(\beta)\sum_{\substack{p,q\\i,j,k}}|\hat{F}_\beta^{ki}(\alpha)_{pq}|^2 \E_{\V s}\left[\hat{h}_u^{\V s}(\beta)_{jk}\right]^2\right)
		\end{align*}
		The second term is bounded by $1$ as
		\begin{align*}
			&\sum_{\alpha,\beta}\dim(\alpha)\dim(\beta)\sum_{\substack{p,q\\i,j,k}}|\hat{F}_\beta^{ki}(\alpha)_{pq}|^2\E_{\V s}\left[\hat{h}_u^{\V s}(\beta)_{jk}\right]^2\\
			&\leq\E_{\V s}\left[ \sum_{\beta}\dim(\beta)\sum_{j,k}|\hat{h}_u^{\V s}(\beta)_{jk}|^2\sum_{i}\sum_{\alpha}\dim(\alpha    )\sum_{p,q}|\hat{F}_\beta^{ki}(\alpha)_{pq}|^2\right]\\
			&=\E_{\V s}\left[ \sum_{\beta}\dim(\beta)\sum_{j,k}|\hat{h}_u^{\V s}(\beta)_{jk}|^2\sum_{i}\sum_{\alpha}\dim(\alpha    )\hsnorm{\hat{F}_\beta^{ki}(\alpha)}^2\right]\\
			&=\E_{\V s}\left[\sum_{\beta}\dim(\beta)\sum_{j,k}|\hat{h}_u^{\V s}(\beta)_{jk}|^2\sum_{i}\|F_\beta^{ki}\|^2\right]\\
			&= \E_{\V s}\left[\sum_{\beta}\dim(\beta)\sum_{j,k}|\hat{h}_u^{\V s}(\beta)_{jk}|^2\right] \tag*{(Using Equation~(\ref{eq:sum_one})}\\
			&\leq \E_{\V s} \left[\pnorm{h_u^{\V s}}{2}^2 \right]= 1,
		\end{align*}
		Based on the above bound, the term $  |\Theta^{e(u,v)}_{\mathrm{Low}}|^2$ is upper bounded by
		\begin{align*}
			|\Theta^{e(u,v)}_{\mathrm{Low}}|^2&\leq \sum_{\substack{\alpha,\beta\\|s(\alpha)|,|s(\beta)|\neq 0\\|s(\beta)|<C\\s(\alpha)\subseteq \pi(s(\beta))}}\dim(\alpha)\dim(\beta)\sum_{\substack{p,q\\i,j,k}}|\hat{f}_v(\alpha)_{pq}|^2|\hat{g}_u(\beta)_{ij}|^2
		\end{align*}
		Since for $\beta$ such that $|s(\beta)| \leq C$, $\dim(\beta)=\prod_{i=1}^L \dim(\beta_i)=\prod_{i,\dim(\beta_i)\geq 2}\dim(\beta_i)\leq(\sqrt{|G|})^{C}$ and the index $i$ varies over the dimension of $\beta$,
		\begin{align*}
			|\Theta^{e(u,v)}_{\mathrm{Low}}|^2&\leq|G|^{\frac{C}{2}}\sum_{\substack{\alpha,\beta\\|s(\alpha)|,|s(\beta)|\neq 0\\|s(\beta)|<C\\s(\alpha)\subseteq \pi(s(\beta))}}\dim(\alpha)\dim(\beta)\sum_{\substack{p,q\\i,j}}|\hat{f}_v(\alpha)_{pq}|^2|\hat{g}_u(\beta)_{ij}|^2\\
			&=|G|^{\frac{C}{2}}\sum_{\substack{\alpha,\beta\\|s(\alpha)|,|s(\beta)|\neq 0\\|s(\beta)|<C\\s(\alpha)\subseteq \pi(s(\beta))}}\dim(\alpha)\dim(\beta)\hsnorm{\hat{f}_v(\alpha)}^2\hsnorm{\hat{g}_u(\beta)}^2.
		\end{align*}
		
		Now, we can present the decoding strategy for a typical edge $e=(u,v)$:
		\paragraph{Decoding strategy.} 
		\begin{enumerate}
			\item For each $u\in U$, consider a function $g_{u}(\V x)=\chi_\rho(A_u(\V x))$, sample a $\beta$ with probability $\dim(\alpha)\hsnorm{\hat{g}_{u}(\beta)}^2$ and select a random coordinate $j$ s.t. $\chi_{\beta_j}\notin\widehat{G/H_S}$. If there is no such $j$, then return $\bot$.
			\item For each $v\in V$, consider a function $f_{v}(\V x)=\chi_\rho(A_v(\V x))$, sample an $\alpha$ with probability $\dim(\alpha)\hsnorm{f_{v}(\alpha)}^2$ and select a random coordinate $i$ s.t. $\chi_{\alpha_i}\notin \widehat{G/H_S}$. If there is no such $i$, then return $\bot$.
		\end{enumerate}
		For $\alpha, \beta$ such that $s(\alpha)$, $s(\beta)$ are nonempty and $s(\alpha)\subseteq \pi(s(\beta))$, the strategy will succeed with probability at least $1/|s(\beta)|$. This is because for any label $\ell$ returned by player $v$, the condition $s(\alpha)\subseteq \pi(s(\beta))$ guarantees that there exists a $\ell'\in\pi^{-1}_e(\ell)$ such that $\beta_{\ell'}\notin \widehat{G/H_S}$, and the player $u$ returns this $\ell'$ with probability $1/|s(\beta)|$. Therefore, the expected value of the labeling returned by the strategy is given by
		\begin{align*}
			&\E_{e(u,v)}\left[\sum_{\substack{\alpha,\beta\\|s(\alpha)|,|s(\beta)|\neq 0\\|s(\beta)|<C\\s(\alpha)\subseteq \pi(s(\beta))}}\dim(\alpha)\dim(\beta)\hsnorm{\hat{f}_v(\alpha)}^2\hsnorm{\hat{g}_u(\beta)}^2\frac{1}{|s(\beta)|}\right]\\
			&\geq \frac{1}{C}\E_{e(u,v)}\left[\sum_{\substack{\alpha,\beta\\|s(\alpha)|,|s(\beta)|\neq 0\\|s(\beta)|<C\\s(\alpha)\subseteq \pi(s(\beta))}}\dim(\alpha)\dim(\beta)\hsnorm{\hat{f}_v(\alpha)}^2\hsnorm{\hat{g}_u(\beta)}^2\right]\\
			&\geq \frac{1}{C\cdot|G|^\frac{C}{2}}\E_{e(u,v)}\left[|\Theta^{e(u,v)}_{\mathrm{Low}}|^2\right]\\
			&\geq \frac{1}{C\cdot|G|^\frac{C}{2}}  \left|\E_{e(u,v)}\left[\Theta^{e(u,v)}_{\mathrm{Low}}\right]\right|^2     \tag*{(Cauchy-Schwarz inequality)}\\
			&\geq \frac{\delta'^2}{4C|G|^{C/2}}.
		\end{align*}
		% We will set the value of $C$\abnote{Set $C$ such that the above is  $>2^{-r}$ to arrive at a contradiction} so that we get the required upper bound of $\left|\E_{e(u,v)} \left[ \Theta^{e(u,v)}_{\mathrm{High}} \right]\right| \leq \delta'/2$.
        We set $C=\Omega_G(r-2\log\delta')$ so that $\frac{\delta'^2}{4C|G|^{C/2}}> O(2^{-r})$, which contradicts that hardness of \LC. Therefore, $|\E_{e(u,v)}[ \Theta^{e(u,v)}_{\mathrm{Low}}]|\leq \delta'/4$ for any two layers in $\mathcal{H}$, and furthermore

        $$\left|\E_{1\leq t < t'\leq T}\left[\E_{\pi_{u,v}\in \Pi_{t,t'}} \left[ \Theta^{e(u,v)}_{\mathrm{Low}} \right]\right]\right| \leq \frac{\delta'}{4}.$$
		
		\paragraph{Bounding the $\Theta^{e(u,v)}_{\mathrm{{High}}}$ term.} It remains to bound $|\E_{e}[\Theta^{e}_{\mathrm{{High}}}]|$. We divide the term into two parts for $D$ such that $C\geq \Omega_{G}\left(\log{D}-\log{\delta'}\right)$.
		\begin{align*}
				\Theta^{e(u,v)}_{\mathrm{High}} &= \sum_{\substack{\alpha,\beta\\|s(\alpha)|,|s(\beta)|\neq 0\\|s(\beta)|\geq C\\s(\alpha)\subseteq \pi(s(\beta))}}\termm{e}{\alpha}{\beta}\\
				& = \sum_{\substack{\alpha,\beta\\|s(\alpha)|,|s(\beta)|\neq 0\\|s(\beta)|\geq C\\s(\alpha)\subseteq \pi(s(\beta)) \\ \dim(\beta)\leq D}}\termm{e}{\alpha}{\beta} +  \sum_{\substack{\alpha,\beta\\|s(\alpha)|,|s(\beta)|\neq 0\\|s(\beta)|\geq C\\s(\alpha)\subseteq \pi(s(\beta))\\ \dim(\beta)> D}}\termm{e}{\alpha}{\beta},
		\end{align*}
	We denote the first term by $\Theta^{e(u,v)}_{\mathrm{High}, \leq D}$ and the second term by $	\Theta^{e(u,v)}_{\mathrm{High}, >D}$.  We bound these terms separately,

	\paragraph{Bounding $\Theta^{e(u,v)}_{\mathrm{High}, \leq D}$.} Starting with the simplified expression for $\termm{e}{\alpha}{\beta}$, we have
	\begin{align*}
		\Theta^{e(u,v)}_{\mathrm{High}, \leq D} = \sum_{\alpha, \beta}\dim(\alpha)\dim(\beta)\sum_{p,q,i,k}\hat{f}_v(\alpha)_{pq}\cdot \E_{\V s}\left[(\hat{g}_u(\beta)\hat{h^{\V s}_u}(\beta))_{ik}\right]\cdot \hat{F_{\beta}^{ki}}(\alpha)_{pq}.
	\end{align*} 
We now simplify the expectation over ${\V s}$. Recall that the function $h_u^{\V s}(\V x) := \chi_\rho(A_u(\V x\gp \V s)) = g_u(\V x\gp \V s)$. We now express the Fourier coefficient of $h_u^{\V s}$ in terms of the Fourier coefficient of $g_u^{\V s}$. By the Fourier inversion formula,
\begin{align*}
	h_u^{\V s} (\V x)  = g_u(\V x\gp \V s) &= \sum_\beta \dim(\beta) \trace(\hat{g}_u(\beta) \beta(\V x\gp \V s)^\star)\\
     &= \sum_\beta \dim(\beta) \trace(\hat{g}_u(\beta) (\beta(\V x)\beta(\V s))^\star) \tag*{(\text{Using homomorphism of $\beta$})}\\
  &= \sum_\beta \dim(\beta) \trace(\hat{g}_u(\beta) \beta(\V s)^\star\beta(\V x)^\star)
\end{align*}
As the Fourier expansion is unique, we have $\hat{h}_u^{\V s}(\beta) = \hat{g}_u(\beta) \beta(\V s)^\star$. Using this, we have
\begin{align*}
	\E_{\V s}\left[(\hat{g}_u(\beta)\hat{h^{\V s}_u}(\beta))_{ik}\right] &= \E_{\V s}\left[(\hat{g}_u(\beta) \hat{g}_u(\beta) \beta( \V s^{-1}))_{ik}\right]\\
	&= \E_{\V s}\left[\sum_{j, j'} \hat{g}_u(\beta)_{ij}\hat{g}_u(\beta)_{jj'} \beta( \V s^{-1})_{j'k}\right]\\
    	&= \sum_{j, j'} \hat{g}_u(\beta)_{ij}\hat{g}_u(\beta)_{jj'}\E_{\V s}\left[ \beta( \V s^{-1})_{j'k}\right]\\
        &= \sum_{j, j'} \hat{g}_u(\beta)_{ij}\hat{g}_u(\beta)_{jj'}\E_{\V s}\left[ \prod_{\ell = 1}^L \beta_\ell(s_\ell^{-1})_{j'_\ell,k_\ell}\right]\\        
        &= \sum_{j, j'} \hat{g}_u(\beta)_{ij}\hat{g}_u(\beta)_{jj'}\prod_{\ell = 1}^L \E_{\V s_\ell}\left[ \beta_\ell(s_\ell^{-1})_{j'_\ell,k_\ell}\right]
\end{align*}
Now, we are in the setting when $\dim(\beta)\leq D$ but $s(\beta)\geq C$. Hence, the number of dimension $\geq 2$ representations in $(\beta_1, \beta_2, \ldots, \beta_L)$ is upper bounded by $\log_2 D$. Thus, there are at least $C-\log_2 D$ coordinates $\ell \in [L]$ such that $\beta_\ell \notin  \widehat{G/H_S}$ and $\dim(\beta_\ell)=1$. For each such $\beta_\ell$ we can apply the following claim.

\begin{claim}
\label{claim:main_S_noise_claim}
    There exists some constant $\eps_G>0$ that depends only on $|G|$ such that $|\E_{s\in S}[\beta(s^{-1})]| \leq 1-\eps_G$ for any $\beta \notin\widehat{G/H_S}$ such that $\dim(\beta) = 1$.
\end{claim}
\begin{proof}
    If there exists an 1-dimensional $\beta'\notin\widehat{G/H_S}$ such that $|\E_{s\in S}[\beta'(s)]|=1$, then for all $s\in S$ $\beta'$ must satisfy $\beta'(s)=c$ for a constant $c$. Consider the subgroup of $\widehat{G/[G,G]}$ generated by $\widehat{G/H_S}\cup\{\beta'\}$. This subgroup is isomorphic to $\widehat{G/Q}$ for some $Q\trianglelefteq G$. According to the duality of $G/[G,G]$, we can recover this subgroup $Q$ as
    $$\{g\in G/[G,G]~|~\beta(g)=1,\forall \beta\in\widehat{G/Q}\}.$$
    A character in $\widehat{G/Q}$ will map $S$ to a constant in $\langle c\rangle\subseteq\widehat{G/Q}$. By the definition of quotient group, $\beta(gQ)=\beta(g)$ for all $\beta\in\widehat{G/Q}$, which is also a constant in $\widehat{G/Q}$. Thus, $S$ must be a subset of some coset of $Q$, which indicates $S\subset G/Q$. Meanwhile, $|\widehat{G/Q}|>|\widehat{G/H_S}|$ as $\widehat{G/Q}$ includes more elements. Then we have 
    $$|G/Q|>|G/H_S|,$$
    which contradicts the fact that $H_S$ is the smallest desired subgroup. Consequently, we conclude the claim.

    In addition, since $\beta$ must map $S$ to at least two distinct complex numbers, 
    $$\left|\E_{s\in S}[\beta(s)]\right|\leq \left|\frac{|G|-1}{|G|}+\frac{1}{|G|}e^{\frac{2\pi i}{|G|}}\right|\leq 1-\eps_G,$$
    for some constant $\eps_G$ only depending on $|G|$.
\end{proof}
Using the above claim, we have
\begin{align*}
	\E_{\V s}\left[(\hat{g}_u(\beta)\hat{h^{\V s}_u}(\beta))_{ik}\right] &= \sum_{j, j'} \hat{g}_u(\beta)_{ij}\hat{g}_u(\beta)_{jj'}\prod_{\ell = 1}^L \E_{\V s_\ell}\left[ \beta_\ell(s_\ell^{-1})_{j'_\ell,k_\ell}\right]\\
	 &\leq (1-\eps_G)^{(C-\log_2 D)} \sum_{j, j'} |\hat{g}_u(\beta)_{ij} |\cdot | \hat{g}_u(\beta)_{jj'}|.
\end{align*}
Plugging this upper bound, we get
	\begin{align*}
	\Theta^{e(u,v)}_{\mathrm{High}, \leq D} &= \sum_{\alpha, \beta}\dim(\alpha)\dim(\beta)\sum_{p,q,i,k}\hat{f}_v(\alpha)_{pq}\cdot \E_{\V s}\left[(\hat{g}_u(\beta)\hat{h^{\V s}_u}(\beta))_{ik}\right]\cdot \hat{F_{\beta}^{ki}}(\alpha)_{pq}\\
	&\leq  (1-\eps_G)^{(C-\log_2 D)} \sum_{\alpha, \beta}\dim(\alpha)\dim(\beta)\sum_{\substack{p,q,i,k\\ j, j'}}|\hat{f}_v(\alpha)_{pq}|\cdot |\hat{g}_u(\beta)_{ij} |\cdot | \hat{g}_u(\beta)_{jj'}|\cdot |\hat{F_{\beta}^{ki}}(\alpha)_{pq}|
\end{align*} 
Applying the Cauchy-Schwarz inequality,
	\begin{align*}
	&|\Theta^{e(u,v)}_{\mathrm{High},\leq   D}|^2\leq (1-\eps_G)^{2(C-\log_2 D)} \left(\sum_{\alpha,\beta}\dim(\alpha)\dim(\beta)\sum_{\substack{p,q, i, k\\ j, j'}}|\hat{f}_v(\alpha)_{pq}|^2|\hat{g}_u(\beta)_{ij}|^2\right)\\
&\quad\quad\quad\quad\quad\quad\quad\quad\quad\quad\left(\sum_{\alpha,\beta}\dim(\alpha)\dim(\beta)\sum_{\substack{p,q, i, k\\j, j'}}|\hat{F}_\beta^{ki}(\alpha)_{pq}|^2 |\hat{g}_u(\beta)_{jj'}|^2\right).
\end{align*}
Using the fact that $i, j , j'$ and $k$ vary over the dimension of $\beta$, which is at most $D$, the first term is at most,
\begin{align*}
	&\left(\sum_{\alpha,\beta}\dim(\alpha)\dim(\beta)\sum_{\substack{p,q, i, k\\ j, j'}}|\hat{f}_v(\alpha)_{pq}|^2|\hat{g}_u(\beta)_{ij}|^2\right)\\
    &\quad\quad\quad\leq D^2 \left(\sum_{\alpha,\beta}\dim(\alpha)\dim(\beta)\sum_{\substack{p,q, i, j}}|\hat{f}_v(\alpha)_{pq}|^2|\hat{g}_u(\beta)_{ij}|^2\right)\\
	&\quad\quad\quad\leq D^2 \left(\sum_{\alpha}\dim(\alpha)\sum_{\substack{p,q}}|\hat{f}_v(\alpha)_{pq}|^2\right) \left(\sum_{\beta}\dim(\beta)\sum_{\substack{i, j}}|\hat{g}_u(\beta)_{ij}|^2\right)\\
	&\quad\quad\quad \leq D^2 \|f_v\|_2^2 \|g_u\|_2^2\leq D^2.
\end{align*}
Similarly, the second term is
\begin{align*}
	&\left(\sum_{\alpha,\beta}\dim(\alpha)\dim(\beta)\sum_{\substack{p,q, i, k\\j, j'}}|\hat{F}_\beta^{ki}(\alpha)_{pq}|^2|\hat{g}_u(\beta)_{jj'}|^2\right) \\
    &\quad\quad\quad\leq D \sum_{\beta}\dim(\beta)\sum_{j, j'}|\hat{h}_u^{\V s}(\beta)_{jj'}|^2\sum_{i}\sum_{\alpha}\dim(\alpha    )\sum_{p,q}|\hat{F}_\beta^{ki}(\alpha)_{pq}|^2\\
	&\quad\quad\quad=D \sum_{\beta}\dim(\beta)\sum_{j, j'}|\hat{g}_u(\beta)_{jj'}|^2\sum_{i}\sum_{\alpha}\dim(\alpha    )\hsnorm{\hat{F}_\beta^{ki}(\alpha)}^2\\
	&\quad\quad\quad=D \sum_{\beta}\dim(\beta)\sum_{jj'}|\hat{g}_u(\beta)_{jj'}|^2\sum_{i}\|F_\beta^{ki}\|^2\\
	&\quad\quad\quad=D \sum_{\beta}\dim(\beta)\sum_{j,j'}|\hat{g}_u(\beta)_{jj'}|^2 \tag*{(Using Equation~(\ref{eq:sum_one})}\\
	&\quad\quad\quad\leq D \cdot \pnorm{g_u}{2}^2 = D,
\end{align*}
Therefore, we have,
$$\left|\E_{e(u,v)}\left[\Theta^{e(u,v)}_{\mathrm{High},\leq   D} \right]\right|^2 \leq \E_{e(u,v)}\left[|\Theta^{e(u,v)}_{\mathrm{High},\leq   D}|^2 \right]\leq  (1-\eps_G)^{2(C-\log_2 D)} D^3.$$
We verify that a setting of $D$ satisfies $(1-\eps_G)^{2(C-\log_2 D)} D^3 \leq \frac{\delta'^2}{16}$:
\begin{align*}
    (1-\eps_G)^{2(C-\log_2 D)} D^3&\leq\frac{\delta'^2}{16}\\
    2(C-\log_2 D)\log{(1-\eps_G)}&\leq2\log{\delta'}-4-3\log D\\
    \implies  C&\geq \Omega_{G}\left(\log{D}-\log{\delta'}\right).
\end{align*}

	\paragraph{Bounding $\Theta^{e(u,v)}_{\mathrm{High}, > D}$.} 
    
 We use the proof of the ~\cite[Claim 4.5]{BhangaleK21} to prove the following.

    \begin{claim}
    \label{claim:large_dimension}
    For every edge $e=(u,v)$, and $C = \Omega_G(\log (1/\delta'))$, we have 
		$$\left|\E_{\V s}\left[\Theta^{e(u,v)}_{\mathrm{High}, > D}\right] \right| \leq \frac{\delta'}{4}.$$
	\end{claim}
	\begin{proof}
		First, using the proof of~\cite[Claim 4.5]{BhangaleK21} and using the fact that $\hat{h}_u^{\V s}(\beta) = \hat{g}_u(\beta) \beta(\V s)^\star$, we have the following upper bound.
        \begin{align*}
            \left|\E_{\V s}[\Theta^{e(u,v)}_{\mathrm{High}, > D}] \right| &\leq \|f_v\|_2^2 \|g_u\|_2^2 \sum_{\beta,  |s(\beta)|\geq C}\dim(\beta) \hsnorm{\hat{g}_u(\beta) \cdot \E_{\V s}[\beta(\V s)^\star]}^2\\
            & = \sum_{\beta,  |s(\beta)|\geq C}\dim(\beta) \hsnorm{\hat{g}_u(\beta) \cdot \E_{\V s}[\beta(\V s)^\star]}^2
        \end{align*}
        We now effectively bound the last summation by bounding $\hsnorm{\hat{g}_u(\beta) \cdot \E_{\V s}[\beta(\V s)^\star]}^2$. We have,
        \begin{align*}
           \hsnorm{\hat{g}_u(\beta) \cdot \E_{\V s}[\beta(\V s)^\star]}^2 &\leq \hsnorm{\hat{g}_u(\beta)}^2 \cdot \opnorm{\E_{\V s}[\beta(\V s)^\star]}^2,
        \end{align*}
       where the inequality follows from the following fact:
       
  \[
\hsnorm{AB}^2=\operatorname{Tr}(B^\star A^\star A B)\le \opnorm{B}^2\,\operatorname{Tr}(A^\star A)=\opnorm{B}^2\,\hsnorm{A}^2.
\]
       We now study the quantity $\opnorm{\E_{\V s}[\beta(\V s)^\star]}^2$ for $\beta$s such that $|s(\beta)|\geq C$. We have

       \begin{align*}
           \opnorm{\E_{\V s}[\beta(\V s)^\star]}^2 &=  \opnorm{\E_{(s_1, s_2, \ldots, s_L)\sim S^L}\left[\otimes_{i=1}^L\beta_i(s_i)^\star\right] }^2\\
            &=  \opnorm{\otimes_{i=1}^L\E_{s\sim S}\beta_i(s)^\star] }^2\\
           & = \prod_{i=1}^L \opnorm{\E_{s\sim S}[\beta_i(s)^\star]}^2,
       \end{align*}
       where we used the fact $\opnorm{A\otimes B} = \opnorm{A}\cdot \opnorm{B}$. As $\beta_i$s are unitary transformations, we have $\opnorm{\beta_i(s)}\leq 1$ for all $i\in [L]$ and $s\in S$ and hence $\opnorm{\E_{s\sim S}[\beta_i(s)^\star]}\leq 1$. We now show that for $i$ such that $i\in s(\beta)$, we have $\opnorm{\E_{s\sim S}[\beta_i(s)^\star]}^2\leq 1-\delta_G$, where $\delta_G>0$ only depends on $|G|$. With this, for $\beta$ such that $|s(\beta)|\geq C$, we have,
       
       $$ \opnorm{\E_{\V s}[\beta(\V s)^\star]}^2 \leq (1-\delta_G)^C.$$
Let us see why this finishes the proof of the claim.

         \begin{align*}
            \left|\E_{\V s}[\Theta^{e(u,v)}_{\mathrm{High}, > D}] \right| &\leq \sum_{\beta, |s(\beta)|\geq C}\dim(\beta) \hsnorm{\hat{g}_u(\beta)}^2 \cdot \opnorm{\E_{\V s}[\beta(\V s)^\star]}^2\\
            & \leq \sum_{\beta, |s(\beta)|\geq C}\dim(\beta) \hsnorm{\hat{g}_u(\beta)}^2 \cdot (1-\delta_G)^C\\
            & = (1-\delta_G)^C \sum_{\beta, |s(\beta)|\geq C}\dim(\beta) \hsnorm{\hat{g}_u(\beta)}^2 \\
            & \leq (1-\delta_G)^C \|g_u\|_2^2 = (1-\delta_G)^C \leq \frac{\delta'}{4}, 
        \end{align*}
        where the last inequality follows from the choice of $C$.
    Thus, it remains to show the following claim.
    \begin{claim}
    \label{claim:highterms}
        For $i$ such that $i\in s(\beta)$, we have $\opnorm{\E_{s\sim S}[\beta_i(s)^\star]}^2\leq 1-\delta_G$ for some $\delta_G>0$.
    \end{claim}
    \begin{proof}
        % \abnote{The proof uses $S^{-1}S$ generates $H_S$ which is always true in the abelian setting. So, for now, our main theorem is true for $S$ such that $S^{-1}S$ generates $H_S$} 
        The claim for $i$ such that $\dim(\beta_i) =1$ follows from Claim~\ref{claim:main_S_noise_claim}. Therefore, we assume that $\dim(\beta_i)>1$. Define
$$
M:=\mathbb{E}_{s\in S}\big[\beta_i(s^{-1})\big]=\frac{1}{|S|}\sum_{s\in S}\beta_i(s^{-1}).
$$
Since each $\beta_i(s^{-1})$ is unitary, $M$ is an average of unitaries, and hence $\opnorm{M}\le 1$. Suppose for contradiction that $\opnorm{M}=1$, then $\opnorm{M^\star}=1$. Then there exists a unit vector $v$ such that $\|M^\star v\|_2=1$.

Write $v_s:=\beta_i(s)v$ for $s\in S$. Then $\|v_s\|_2=1$ for all $s$, and
$$
M^\star v=\frac{1}{|S|}\sum_{s\in S} v_s.
$$
By the triangle inequality,
$$
1=\|M^\star v\|_2=\left\|\frac{1}{|S|}\sum_{s\in S} v_s\right\|_2
\le \frac{1}{|S|}\sum_{s\in S}\|v_s\|_2
=1.
$$
Hence equality holds in the triangle inequality, which implies all vectors $v_s$ have the same direction; since they all have norm $1$, we get $v_s=v_t$ for all $s,t\in S$. Therefore, for all $s,t\in S$,
$$
\beta_i(s^{-1}t)v=\beta_i(s^{-1})\beta_i(t)v=v,
$$
so $v$ is fixed by every element of $S^{-1}S$, and hence by the subgroup it generates:
$$
\beta_i(h)v=v \qquad \forall\, h\in \langle S^{-1}S\rangle = H_S.
$$
In particular, $v$ is fixed by $[G,G]\subseteq H_S$. Let
$$
V^{[G,G]}:=\{w\in \mathbb{C}^d:\beta_i(c)w=w\ \forall\, c\in [G,G]\}.
$$
Because $[G,G]\lhd G$, the subspace $V^{[G,G]}$ is $G$-invariant: for $w\in V^{[G,G]}$, $g\in G$, and $c\in [G,G]$,
$$
\beta_i(c)\beta_i(g)w=\beta_i(g)\beta_i(g^{-1}cg)w=\beta_i(g)w,
$$
since $g^{-1}cg\in [G,G]$. Thus $\beta_i(g)w\in V^{[G,G]}$. We have exhibited a nonzero vector $v\in V^{[G,G]}$, so by irreducibility of $\beta_i$ we must have $V^{[G,G]}=\mathbb{C}^d$. Hence $\beta_i(c)=I$ for all $c\in [G,G]$, i.e. $[G,G]\subseteq \ker(\beta_i)$.

Therefore $\beta_i$ factors through the abelian quotient $G/[G,G]$. But every irreducible complex representation of an abelian group is $1$-dimensional, contradicting $d>1$. This contradiction shows that $\opnorm{M}\neq 1$, and since $\opnorm{M}\le 1$, we conclude $\opnorm{M}\leq 1-\delta_G$ for some $\delta_G>0$.
    \end{proof}
    This finishes the proof of the Claim~\ref{claim:large_dimension}.
	\end{proof}

\paragraph{Finishing the proof.} Using these bounds, we have,
\begin{align*}
\left|\E_{e(u,v)}\left[\Theta^{e(u,v)}_{\mathrm{High}} \right]\right| = \left|\E_{e(u,v)}\left[\Theta^{e(u,v)}_{\mathrm{High}, \leq D} + \Theta^{e(u,v)}_{\mathrm{High}, >D} \right]\right|&\leq \left|\E_{e(u,v)}\left[\Theta^{e(u,v)}_{\mathrm{High}, \leq D} \right]\right| + \left|\E_{e(u,v)}\left[\Theta^{e(u,v)}_{\mathrm{High}, >D} \right]\right|\\
&\leq \frac{\delta'}{4}+\frac{\delta'}{4} = \frac{\delta'}{2} \leq \frac{\delta}{4|G|},  
\end{align*}
as required.
\end{proof}
Therefore, the term (\ref{eq:nonconstant}) collectively can be upper bounded by $\frac{\delta}{2}$.
 
	\subsubsection{Bounding expressions in (\ref{eq:greaterdim}).} Let us simplify term~(\ref{eq:greaterdim}),
    \begin{align*}
	&\frac{1}{|G|}\sum_{\dim(\rho)\geq 2}\sum_{s\in S}\E_{\substack{\pi_{uv}\in \Pi_{t,t'}\\\V x ,\V y, \V s}}\left[ \dim(\rho)\cdot \chi_{\rho}(A_v(\V x)\gp A_{u}(\V y)\gp  A_{u}(\V z)\gp s^{-1}) \right]\\
        &\leq\sum_{\dim(\rho)\geq 2}\sum_{s\in S}\E_{\substack{\pi_{uv}\in \Pi_{t,t'}\\\V x ,\V y, \V s}}\left[ \cdot \chi_{\rho}(A_v(\V x)\gp A_{u}(\V y)\gp  A_{u}(\V z)\gp s^{-1}) \right].
    \end{align*}
    
The $\dim(\rho)\geq 2$ case also appears in the Max-E$3$-Lin proof in \cite{BhangaleK21}, but we have additional $\V s\in S^{R_t}$ and $s\in S$ terms; We focus on dealing with these terms. The expectation in the claim is
\begin{align*}
    &\E_{\pi_{uv}\in\Pi_{t,t'}}\E_{\V x,\V y,\V z}[\chi_\rho(A_v(\V x)\gp A_u(\V y)\gp A_u(\V z)\gp s^{-1})]\\
    &=\sum_{\dim(\rho)\geq 2}\sum_{s\in S}\E_{\substack{\pi_{uv}\in\Pi_{t,t'}\\\V x ,\V y, \V s}}\left[ \trace{(\rho(A_v(\V x)\gp A_{u}(\V y)\gp  A_{u}(\V z))\cdot \rho(s^{-1}))} \right].
	&\tag{$\rho$ is a homomorphism.}
	\end{align*}
    Fix an edge $e(u,v)$, a representation $\rho$ and $s\in S$, we have,
	\begin{align*}
		&\E_{\V x,\V y,\V s}\left[\trace{(\rho(A_v(\V x)\gp A_{u}(\V y)\gp  A_{u}(\V z))\cdot \rho(s^{-1}))} \right]\\
		=&\E_{\V x,\V y,\V s}\left[\trace{(\rho(A_v(\V x))\cdot \rho(A_{u}(\V y))\cdot  \rho(A_{u}(\V z))\cdot \rho(s^{-1}))} \right]\\
		=&\E_{\V x,\V y,\V s}\left[\sum_{1\leq p,q,r,w\leq \dim(\rho)}\rho(A_v(\V x))_{pq}\cdot \rho(A_{u}(\V y))_{qr}\cdot  \rho(A_{u}(\V z))_{rw}\cdot \rho(s^{-1})_{wq}\right]\\
		=&\sum_{1\leq p,q,r,w\leq \dim(\rho)}\E_{\V x,\V y,\V s}\left[\rho(A_v(\V x))_{pq}\cdot \rho(A_{u}(\V y))_{qr}\cdot  \rho(A_{u}(\V z))_{rw}\cdot \rho(s^{-1})_{wq}\right].
	\end{align*}
    Let $f_{pq}(\V x):=\rho(A_v(\V x))_{pq}$, $g_{qr}(\V y)=\rho(A_{u}(\V y))_{qr}$ and $h^{\V s}_{rw}(\V x)=\rho(A_u(\V x))_{rw}$. Fix some $p,q,r,w$, $\rho(s^{-1})_{wq}$ is an entry of $\rho(s^{-1})$ bounded by $1$. Hence, if we fix a vector $\V s$, for $\rho\in \irr(G)$ with $\dim(\rho)\geq 2$, 
	\begin{align*}
    &\E_{\V x,\V y}\left[\rho(A_v(\V x))_{pq}\cdot \rho(A_{u}(\V y))_{qr}\cdot  \rho(A_{u}(\V z))_{rw}\cdot \rho(s^{-1})_{wq}\right]\\
    &\leq \E_{\V x, \V y}\left[\rho(A_v(\V x))_{pq}\cdot \rho(A_{u}(\V y))_{qr}\cdot  \rho(A_{u}(\V z))_{rw}\right]\\
    &=\E_{\V x, \V y} \left[f_{pq}(\V x)\cdot g_{qr}(\V y) \cdot h_{rw}^{\V s}(\V y^{-1}\gp\V (\V x\circ \pi)^{-1}\V)\right]\\
    &=\E_{\V x}[f_{pq}(\V x)\cdot (g_{qr}* h^{\V s}_{rw})((\V x\circ \pi)^{-1}))]\\
    &=\E_{\V x}\left[\left(\sum_{\beta}\dim(\beta)\trace{(\hat{g}(\beta)\hat{h}(\beta)\beta(\V x \circ\pi))}\right)\left(\sum_{\alpha}\dim(\alpha)\trace{(\hat{f}(\alpha)\alpha(\V x^{-1}))}\right)\right]\\
    &=\sum_{\substack{\alpha,\beta\\\dim(\alpha),\dim(\beta)\geq 2}}\dim(\alpha)\dim(\beta)\E_{\V x}\left[\trace{(\hat{g}(\beta)\hat{h}(\beta)\beta(\V x \circ\pi))}\cdot \trace{(\hat{f}(\alpha)\alpha(\V x^{-1}))} \right].
\end{align*}
The last steps follow Lemma~\ref{lemma:folded_zero_coeff}, which states $\hat{f}(\alpha)=0$ for any $\dim(\alpha)=1$ if $f(\V x)=\rho(A_{v}(\V x))$ is folded and $\dim(\rho)\geq 2$. Similar for functions $g(\V x)$ and $h(\V x)$. 

Again, we can rewrite this summation as
\begin{align*}
    &\sum_{\substack{\alpha,\beta\\\dim(\alpha),\dim(\beta)\geq 2}}\dim(\alpha)\dim(\beta)\E_{\V x}\left[\trace{(\hat{g}(\beta)\hat{h}(\beta)\beta(\V x \circ\pi))}\cdot \trace{(\hat{f}(\alpha)\alpha(\V x^{-1}))} \right]\\
    &=\underbrace{\sum_{\substack{\alpha,\beta\\\dim(\alpha),\dim(\beta)\geq 2\\\dim_{\geq 2}(\beta)<K}}\dim(\alpha)\dim(\beta)\E_{\V x}\left[\trace{(\hat{g}(\beta)\hat{h}(\beta)\beta(\V x \circ\pi))}\cdot \trace{(\hat{f}(\alpha)\alpha(\V x^{-1}))} \right]}_{\Gamma_{\mathrm{low}}^{p,q,r,w}}\\
    &+\underbrace{\sum_{\substack{\alpha,\beta\\\dim(\alpha),\dim(\beta)\geq 2\\\dim_{\geq 2}(\beta)\geq K}}\dim(\alpha)\dim(\beta)\E_{\V x}\left[\trace{(\hat{g}(\beta)\hat{h}(\beta)\beta(\V x \circ\pi))}\cdot \trace{(\hat{f}(\alpha)\alpha(\V x^{-1}))} \right]}_{\Gamma_{\mathrm{high}}^{p,q,r,w}},
\end{align*}
where $\dim_{\geq 2}(\beta)$ denotes the number of representations in $\beta=(\rho_1,\rho_2,\dots,\rho_L)$ which are of dimension at least $2$. $K$ is a constant adopted from~\cite{BhangaleK21new} to divide the low-degree and high-degree terms, which uses different settings regarding the \LLC. We have the following lemma:
    \begin{lemma}
    \label{lemma:highterm_bound}
    Letting $K:=\frac{8|G|^{10}}{\delta^2}$. If $\mathcal{H}$ is at most $\min(\frac{\delta^2}{10|G|^{10K}}, \frac{\delta^2}{2|G|^{20}})$-satisfiable, then 

    $$\left|\E_{1\leq t<t'\leq T}\left[\E_{\pi_{uv}\in\Pi_{t,t'}}\left[\E_{\V x,\V y,\V z}[\chi_\rho(A_v(\V x)\gp A_u(\V y)\gp A_u(\V z)\gp s^{-1})]\right]\right]\right|\leq \frac{\delta}{|G|^{3}}.$$
    \end{lemma}
   
\begin{proof}

Now, for a fixed $\V s$, $h^{\V s}$ is still a folded function as $h(c\V x)=g(c\V x\gp \V s^{-1})=c\cdot g(\V x\gp\V s^{-1})=c\cdot h(\V x)$. Therefore, we can directly apply the following two lemmas from~\cite{BhangaleK21new} to bound $\Gamma_{\high}$ and $\Gamma_{\low}$.
\begin{lemma}(\cite[Claim 4.4]{BhangaleK21new})
If If the \LLC\ instance $\mathcal{H}$ is at most $\frac{\delta^2}{10|G|^{10K}}$-satisfiable, then for every $1\le t< t'\le T$, every $\rho\in \irr(G)$ such that $\dim(\rho)\geq 2$ and every $1\leq p,q,r,w\leq \dim(\rho)$,
	$$\left| \E_{\pi_{uv}\in\Pi_{t,t'}}[\Gamma^{p,q,r,w}_{\low}] \right|\leq \frac{\delta}{2|G|^{5}}.$$
\end{lemma}

This lemma guarantees that the expectation over the low-degree term $\Gamma_{\mathrm{low}}^{p,q,r,w}$ is negligible between any two layers. Therefore,

$$\left| \E_{1\leq t< t'\leq T}\left[\E_{\pi_{uv}\in \Pi_{t,t'}}[\Gamma^{p,q,r,w}_\low]\right] \right|\leq \frac{\delta}{2|G|^5}.$$

\begin{lemma}(\cite[Claim 4.5]{BhangaleK21new})
\label{lemma:bk_new_high}
    If the \LLC\ instance $\mathcal{H}$ is at most $\frac{\delta^2}{(2|G|)^{20}}$-satisfiable, then for every $\rho\in \irr(G)$ such that $\dim(\rho)\geq 2$ and for every $1\leq p,q,r,w\leq \dim(\rho)$,
	$$\left| \E_{1\leq t< t'\leq T}\left[\E_{\pi_{uv}\in \Pi_{t,t'}}[\Gamma^{p,q,r,w}_\high]\right] \right|\leq \frac{\delta}{2|G|^5}.$$
\end{lemma}

Therefore, the term in Lemma~\ref{lemma:highterm_bound} is at most
$$\sum_{p,q,r,w\leq \dim(\rho)}\frac{\delta}{|G|^5}\leq\sum_{p,q,r,w\leq\sqrt{|G|}} \frac{\delta}{|G|^5} =\frac{\delta}{|G|^3},$$
which finishes the proof.
\end{proof}
Thus, term~(\ref{eq:greaterdim}) is collectively upper bounded by is at most $\frac{\delta}{2}$. Combining this with the bound of $\frac{\delta}{2}$ on term~(\ref{eq:nonconstant}), if the \LLC\ instance is at most $2^{-r}$-satisfiable, then $\textrm{value}(A)\leq |S|/|H_S|+\delta$. This proves Lemma~\ref{lemma:main_soundness}.
\end{proof}
\paragraph{Acknowledgments.} We thank anonymous reviewers for providing useful feedback to improve the presentation of the paper.

\bibliographystyle{alpha}
\bibliography{refs}

\end{document}